\pdfoutput=1
\documentclass[smallextended]{svjour3}

\smartqed  
\usepackage{amsmath}
\usepackage{amssymb}
\usepackage{cases}
\usepackage{enumerate}
\usepackage{color}
\usepackage{graphicx}
\usepackage{epsfig,epstopdf}
\usepackage{hyperref}
\usepackage{subcaption}

\catcode`\@=11
\@addtoreset{equation}{section}   

\@addtoreset{figure}{section}
\renewcommand\thefigure{\thesection.\@arabic\c@figure}

\@addtoreset{table}{section}
\renewcommand\thetable{\thesection.\@arabic\c@table}

\@addtoreset{theorem}{section}

\@addtoreset{remark}{section}

\@addtoreset{lemma}{section}

\renewcommand\theexample{\thesection.\@arabic\c@example}
\@addtoreset{example}{section}

\newcommand{\bz}{{\bf z} }

\newcommand{\be}{\begin{equation}}
\newcommand{\ee}{\end{equation}}
\newcommand{\ba}{\begin{array}}
\newcommand{\ea}{\end{array}}
\newcommand{\bea}{\begin{eqnarray}}
\newcommand{\eea}{\end{eqnarray}}
\newcommand{\beas}{\begin{eqnarray*}}
\newcommand{\eeas}{\end{eqnarray*}}

 \newcommand{\bx}{{\bf x} }

\newcommand{\tphi}{{\tilde{\phi}} }

\newcommand{\bphi}{{\bar{\phi}} }
\newcommand{\brho}{{\bar{\rho}} }

\DeclareMathOperator*{\argmin}{arg\,min}
\newcommand{\tPhi}{\tilde{\Phi}}
\newcommand{\hphi}{\hat{\phi}}
\newcommand{\hrho}{\hat{\rho}}
\newcommand{\tA}{\tilde{A}}
\newcommand{\tF}{\tilde{F}}
\newcommand{\wfig}{7}
\newcommand{\wwfig}{6}
\newcommand{\hfig}{5}

\begin{document}

\title{A normalized gradient flow method with attractive-repulsive splitting for computing ground states of Bose-Einstein condensates with higher-order interaction}

\author{Xinran Ruan}
\institute{
X. Ruan \at
Laboratoire J.-L. Lions, \\
Universit\'{e} Pierre et Marie Curie,\\
 75252 Paris cedex 05, France. 
\email{ruan@ljll.math.upmc.fr}
}

\date{Received: date / Accepted: date}

\maketitle

\begin{abstract}
In this paper, we generalize the normalized gradient flow method to compute the ground states of Bose-Einstein condensates (BEC) with higher order interactions (HOI), which is modelled via the modified Gross-Pitaevskii equation (MGPE). 
Schemes constructed in naive ways suffer from severe stability problems due to the high restrictions on time steps. 
To build an efficient and stable scheme,  
we split the HOI term into two parts with each part treated separately. 
The part corresponding to a repulsive/positive energy is treated semi-implicitly while the one corresponding to an attractive/negative energy is treated fully explicitly. 
Based on the splitting, we construct the BEFD-splitting and BESP-splitting schemes. 
A variety of numerical experiments  shows that the splitting will improve the stability of the schemes significantly. 
Besides, we will show that the methods can be applied to multidimensional problems and to the computation of the first excited state as well.


\end{abstract}

\section{Introduction}\setcounter{equation}{0}

The Bose-Einstein condensate (BEC), which is a many body system with low density and low temperature, has drawn great attention since its first experimental realization in 1995 \cite{Anderson,Bradley,Davis} as it offers a way to measure the microscopic quantum mechanical properties in a macroscopic scale. 
The Gross-Pitaevskii equation (GPE), which is a  mean field approximation by approximating the interaction between particles by an external pseudo-potential \cite{Gross,TGmath,Lieb,Pitaevskii,PitaevskiiStringari}, 
has gained considerable research interest due to its simplicity and effectiveness in describing Bose-Einstein condensates (BEC). 
One key assumption in deriving GPE is that the interaction between particles can be well approximated by the binary interaction in the form 
\be\label{eq:int_gpe}
V_{\rm{int}}(\bx_1-\bx_2)=g_0\delta(\bx_1-\bx_2), \quad \bx_1,\bx_2\in\mathbb{R}^3, 
\ee
where $\delta(\cdot)$ is the Dirac delta function and $g_0=\frac{4\pi \hbar^2 a_{s}}{m}$ is the
contact interaction strength with $a_s$ being the $s$-wave scattering
length, $\hbar$ being the reduced Planck constant and $m$ being the mass of the particle \cite{TGmath}. 
The theory has shown excellent agreement with most experiments.  
However, the validity of the approximation needs to be carefully examined in certain cases, such as in the experiments which take advantage of the Feshbach resonances in cold atomic collision \cite{Zinner}. 
In such cases, higher order interaction (HOI) (or effective range expansion) as a correction to the Dirac delta function has to be taken into account.
In \cite{Collin,Esry}, the higher order interaction correction is analyzed and a new binary interaction is derived as 
\be\label{eq:int}
V_{\rm{int}}(\bz)=g_0\left[\delta(\bz)+\frac{g_1}{2}\left(\delta(\bz)\nabla^2_{\bz}+\nabla^2_{\bz}\delta(\bz)\right)\right],
\ee
where $g_0$ is defined as before, $\bz=\bx_1-\bx_2\in\mathbb{R}^3$ and the HOI correction is given by the parameter
$g_1=\frac{a_s^2}{3}-\frac{a_sr_e}{2}$ with $r_e$ being the effective range of the two-body interaction. 
When $r_e=\frac{2}{3}a_s$, it is for the hard sphere potential and reduces back to the classical case.
In certain cases, $g_1$ can be extremely large \cite{Zinner} and, therefore, the HOI can no longer be ignored. 
With this new choice of the binary interaction \eqref{eq:int}, 
the modified Gross-Pitaveskii equation (MGPE)\cite{JJG,Fu,Collin,Ruan_thesis,Veksler} is derived as 
\be\label{eq:mgpe_origin}
 i\hbar\partial_t \psi=\left[-\frac{\hbar^2}{2m}\nabla^2
+V(\bx)
+Ng_0\left(|\psi|^2+\frac{g_1}{2}\nabla^2|\psi|^2\right)\right]\psi,\, t\ge0,\,\bx\in \mathbb R^3
 \ee
where $N$ is the number of particles,
$V(\bx)$ is a real-valued external trapping potential and $\|\psi(\bx,t)\|_2=1$.

In experiments, the confinement induced by the external potential might be strong in one or two directions. 
As a result, the BEC in 3D could be well described by the MGPE in 2D or 1D, respectively, by performing a proper dimension reduction \cite{Ruan_thesis,BJP,mgpe-asym}. 
Finally, we get the dimensionless modified GPE (MGPE) in $d$-dimensions ($d=1,2,3$) as 
\be\label{mgpe}
 i\partial_t \psi=\left[-\frac{1}{2}\Delta
+V(\bx)
+\beta|\psi|^2-\delta\Delta(|\psi|^2)\right]\psi,\quad t\ge0,\,\bx\in\mathbb{R}^d,
 \ee
 with mass $N(t):=\int_{{\mathbb R}^d}|\psi(\bx,t)|^2d\bx$ and energy 
\be \label{energy}
E(\psi(\cdot,t)):=\int_{{\mathbb R}^d}\biggl[
\frac12|\nabla\psi|^2+V(\bx)|\psi|^2+\frac{\beta}{2}|\psi|^4+\frac{\delta}{2}|\nabla|\psi|^2|^2\biggl]\,d\bx.
\ee
It is easy to check that the mass and the energy are conserved, i.e. 
 \be
N(t)\equiv N(0)=1, \quad E(\psi(\cdot,t))\equiv E(\psi(\cdot,0)). 
\ee

A fundamental problem in studying BEC is to find its stationary
states, especially the ground state which is the stationary state with the lowest energy. 
Mathematically speaking, the ground state $\phi_g^{\beta,\delta}:=\phi_g^{\beta,\delta}(\bx)$ of the MGPE (\ref{mgpe})
is defined as the minimizer of the energy functional (\ref{energy}) under the normalization
constraint, i.e.
\begin{equation}\label{eq:minimize_mgpe}
    \phi_g ^{\beta,\delta}:=  \argmin_{\phi \in S}
    E\left(\phi\right),
  \end{equation}
where $S$ is  defined as 
\be\label{eq:nonconset}
S:=\left\{\phi \, | \,\|\phi\|_2=1,\quad E(\phi)<\infty\right\}.
\ee
$E_g^{\beta,\delta}:=E(\phi_g^{\beta,\delta})$ is called the ground state energy. 
The Lagrangian of the problem \eqref{eq:minimize_mgpe} implies that 
the ground state $\phi_g^{\beta,\delta}$ satisfies the following nonlinear eigenvalue problem
\be\label{eq:eig_mgpe}
\mu \phi=\left[-\frac{1}{2}\Delta+V(\bx)+\beta|\phi|^2-\delta\Delta(|\phi|^2)\right]\phi,
\ee
where the corresponding eigenvalue (also named chemical potential)  $\mu$ can be computed as
\begin{align} 
 \label{eq:mu_mgpe}
&\mu=\int_{{\mathbb R}^d}\biggl[
\frac12|\nabla\phi|^2+V(\bx)|\phi|^2+\beta|\phi|^4+\delta\left|\nabla|\phi|^2\right|^2\biggl]\,d\bx.
\end{align}

It is worth noticing that, when $\delta\neq0$, the ground state exists if and only if $\delta>0$ \cite{mgpe-th}. Furthermore, the ground state is unique if we have both $\beta>0$ and $\delta>0$ \cite{mgpe-th}.
When $\delta=0$, the MGPE degenerates to the GPE. And the existence and uniqueness of the ground state has been thoroughly studied and we refer the readers to \cite{Bao2014,Bao2013,PitaevskiiStringari}. 
 Therefore, throughout the paper, we will only consider the case $\delta\ge0$ for the computation of the ground state of  MGPE.

Numerous numerical methods have been proposed to compute the ground state of the classical GPE, such as a Runge-Kutta spectral method with spectral discretization in space and Runge-Kutta type integration in time by Adhikari et al. in \cite{Muruganandam}, Gauss-Seidel-type methods in \cite{Chang} by Lin et al., a finite element method by directly minimizing the energy functional in \cite{BaoT} by Bao and Tang,
a regularized newton method by Wu, Wen and Bao in \cite{BaoWuWen}, 
a preconditioned nonlinear conjugate gradient method \cite{Antoine} and an adaptive finite element method \cite{Danaila} for the rotating BEC  , 
and so on. 
Among all the methods, the normalized gradient flow method, also named the imaginary time method in physics literatures \cite{Wz1,Bao2004,Chiofalo}, is extremely efficient and easy to implement.  
A Matlab toolbox named GPELab has been developed based on the method \cite{GPELab}. 
Due to its simplicity and efficiency, the method has  been generalized as well to spin-1 BEC \cite{spin_GPE,spin}, rotational BEC \cite{rotation_GPE,Krylov} and so on. 
It seems that the generalization of the method to MGPE \eqref{eq:eig_mgpe} would be trivial. However, it turns out that the term $\delta\Delta(|\psi|^2)\psi$ needs to be carefully dealt with since schemes constructed in naive ways suffer from severe stability problems.  
In this paper, we surprisingly find that a proper splitting of the term $\delta\Delta(|\psi|^2)\psi$ into two parts based on the attractive-repulsive splitting will overcome the problem. The details will be introduced later in Section \ref{subsec:gradientflow_splitting}.

The paper is organized as follows. In Section \ref{sec:NGF}, we introduce the continuous normalized gradient flow method for MGPE and show its discretization, which will be shown to be mass-conserved and energy diminishing. 
Then we introduce the discrete normalized gradient flow and derive the BEFD schemes constructed in naive ways. 
Analysis will be provided to show the high restrictions of BEFD on time steps. 
In Section \ref{sec:splitting}, we introduce the attractive-replulsive splitting of the term $\delta\Delta(|\psi|^2)\psi$ and construct the BEFD-splitting and BESP-splitting schemes by applying finite difference  and pseudo-spectral discretizations in space, respectively.
In Section \ref{sec:numeric}, a variety of numerical experiments will be performed to show that the splitting does improve the stability of the schemes significantly and the BEFD-splitting/BESP-splitting schemes have a great advantage over the methods constructed in naive ways. 
We will also apply the BEFD-splitting/BESP-splitting scheme to multidimensional problems as well as the computation of the first excited state. Finally, some conclusions are drawn in Section \ref{sec:conclusion}. 

\section{Normalized gradient flow and its discretization}\label{sec:NGF} 
The normalized gradient flow method is proven to be one of the most popular methods for computing the ground state of GPE due to its simplicity and efficiency. 
In this section, we will introduce the continuous normalized gradient flow as well as the discrete normalized gradient flow, showing its generalization to MGPE with detailed discretizations. 
\subsection{Continuous normalized gradient flow (CNGF)}

The continuous normalized gradient flow method (CNGF) can be viewed as applying the steepest descent method to the energy functional \eqref{energy} with a Lagrange multiplier for the normalization constraint.  
It is proposed in \cite{Wz1} for computing the ground state of GPE 
 and can be generalized trivially to MGPE in the continuous level as  
\begin{align}
&\phi_t=\frac{1}{2}\Delta\phi-V(\bx)\phi-\beta|\phi|^2\phi+\delta\Delta(|\phi|^2)\phi+\mu_{\phi}(t)\phi, \label{CNGF1}\\
&\phi(\bx,t)=0 \text{ for } \bx\in\partial\Omega,\,\, t\ge0, \quad\text{ and }\quad\phi(\bx,0)=\phi_0(\bx),\quad \bx\in\Omega,
\end{align}
where $\Omega$ is the domain and $\mu_{\phi}(t)$ depending on $\phi:=\phi(\cdot,t)$ is defined as
\be\label{CNGF2}
\mu_{\phi}(t)=\frac{1}{\|\phi\|^2}\int_{\Omega}\left[\frac{1}{2}|\nabla\phi|^2+V(\bx)|\phi|^2+\beta|\phi|^4+\delta|\nabla(|\phi|^2)|^2\right]\,d\bx.
\ee
It can be shown that the CNGF scheme \eqref{CNGF1} conserves the mass and diminishes energy as stated in Lemma \ref{lem:CNGF-cts}. 
\begin{lemma}\label{lem:CNGF-cts}
Denote $\phi(\cdot,t)$ to be the solution of the CNGF scheme \eqref{CNGF1}-\eqref{CNGF2} at time $t$. 
For any initial value $\phi_0(\bx)$ satisfying $\|\phi_0\|=1$ and $\lim_{|\bx|\to\infty}\phi_0(\bx)=0$, we have the
CNGF scheme is normalization conservation and energy diminishing. To be more specific, we have 
\begin{align}
&\|\phi(\cdot,t)\|_2 = \|\phi_0\|_2,\label{CNGF:mass}\\
&\frac{d}{dt}E(\phi(\cdot,t))=-2\|\phi_t(\cdot,t)\|_2^2,\quad t\ge0.\label{CNGF:energy}
\end{align}
which implies that
\be
E(\phi(\cdot,t_1))\ge E(\phi(\cdot,t_2)),\quad 0\le t_1\le t_2<\infty.
\ee
\end{lemma}
\begin{proof}
\eqref{CNGF:mass} and \eqref{CNGF:energy} can be derived by taking  time derivatives of $\|\phi(\cdot,t)\|^2$ and $E(\phi)$, respectively, and combining \eqref{CNGF1}. 
The details are omitted here for brevity. \qed
\end{proof}

\subsection{A mass conserved and energy diminishing discretization}\label{subsec:CNGF}
In this section, we will present the CNGF-FD scheme, which is the full discretization of the CNGF \eqref{CNGF1}-\eqref{CNGF2} via Crank-Nicolson in time and finite difference in space.  

Due to the fact that the external potential $V(\bx)$ satisfies the confining condition, the ground state decays to zero exponentially fast as $|\bx|\to\infty$ \cite{mgpe-th}. 
Therefore, we can always truncate the problem into a large bounded domain with homogeneous Dirichlet boundary conditions in practical computation.  
For simplicity, only the 1D case, which is defined over an interval $(a,b)$, is considered. Extension to higher dimensions is straightforward for  tensor product grids and  thus omitted here.

Choose a time step $\tau>0$ and denote the time sequence as 
$t_n=n\tau$ for $n\ge0$.  
Take $\Omega = (a, b)$ to be the computational domain and denote the uniformly distributed grid points as
\be\label{notation:x}
x_j := a + jh, \text{ for } j = 0, 1, \dots N, 
\ee
where $h := (b - a)/N$ is the mesh size. 
Let $\phi_j^n$ be the numerical approximation of $\phi(x_j,t_n)$ and $V_j:=V(x_j)$.  
 Denote $\Phi^n$ to be the vector solution with component $\phi_j^n$, i.e. 
 \be\label{notation:phi}
 \Phi^n := (\phi_1^n, \dots , \phi_{N-1}^n)^T\in\mathbb{R}^{N-1}.
 \ee  
Define the operators $\delta_x^+$ and $\delta_x^2$ as 
$$\delta_x^+\phi_j^n:=\frac{\phi_{j+1}^n-\phi_j^n}{h}, \quad \delta_{x}^2\phi_j^n:=\frac{\phi_{j+1}^n-2\phi_j^n+\phi_{j-1}^n}{h^2}$$ 
to be the finite difference approximations of  $\partial_x$ and $\partial_{xx}$, respectively. 

With the notations above, the CNGF scheme \eqref{CNGF1}-\eqref{CNGF2}  is discretized, 
for $j=1,2,\dots,N-1$ and $n\ge0$, as   
\begin{align}
&\frac{\phi^{n+1}_j-\phi^n_j}{\tau}=\frac{1}{2}\delta_x^2\bphi^{n}_j-V_j\bphi^{n}_j-\beta\brho^{n}_j\bphi^{n}_j+\delta\delta_x^2\brho^{n}_j\bphi^{n}_j+\bar\mu^{n}\bphi^{n}_j, \label{CNGF-FD}\\
& \phi^0_j=\phi_0(x_j), \qquad \phi_0^n=\phi_N^n=0
\end{align}
where $\bphi^{n}_j:=(\phi^n_j+\phi^{n+1}_j)/2$, $\rho^n_j:=|\phi^n_j|^2$, $\brho^{n}_j:=(\rho^n_j+\rho^{n+1}_j)/2$ and 
\be\label{CNGF-FD2}
\bar\mu^{n}:=\frac{\sum_{j=0}^{N-1}\left[\frac{1}{2}|\delta_x^+\bphi^{n}_j|^2+V_j|\bphi^{n}_j|^2+\beta\brho^{n}_j|\bphi^{n}_j|^2+\delta\delta_x^+\brho^{n}_j\delta_x^+(|\bphi^{n}_j|^2)\right]}{\sum_{j=0}^{N}|\bphi^{n}_j|^2}.
\ee
The  discretized mass and the discretized energy at $t=t_n$ can be computed, respectively, as 
$
\|\Phi^n\|_h^2:=h\sum_{j=1}^{N-1}|\phi^n_j|^2
$
and 
\be\label{E-FD}
E_h(\Phi^n):=h\sum_{j=1}^{N-1}\left[\frac{1}{2}|\delta_x^+\phi^{n}_j|^2+V(x_j)|\phi^{n}_j|^2+\frac{\beta}{2}|\phi^{n}_j|^4+\frac{\delta}{2}|\delta_x^+(|\phi^{n}_j|^2)|^2\right].
\ee

Similar to the CNGF \eqref{CNGF1}-\eqref{CNGF2},  the CNGF-FD scheme \eqref{CNGF-FD}-\eqref{CNGF-FD2} is mass conservation and energy diminishing as well, as shown in Lemma \ref{lem:CNGF}. 
\begin{lemma}\label{lem:CNGF}
The CNGF-FD scheme \eqref{CNGF-FD}-\eqref{CNGF-FD2}  is normalization conservation and energy diminishing, i.e. 
\be\label{lem:CNGF_prop}
\|\Phi^{n+1}\|_h^2=\|\Phi^n\|_h^2,\quad E_h(\Phi^{n+1})\le E_h(\Phi^n), \text{ for all }  n\ge0. 
\ee
\end{lemma}
\begin{proof}
Multiplying $\phi^{n+1/2}_j$ on both sides of \eqref{CNGF-FD} and summing over $j=1,\dots,N-1$, we have 
\be\nonumber
\frac{\|\Phi^{n+1}\|_h^2-\|\Phi^n\|_h^2}{2\tau}=0, 
\ee
which implies the mass conservation. 
 Similarly, multiplying both sides of \eqref{CNGF-FD} by $(\phi_j^{n+1}-\phi_j^n)$  and summing all together, we will get
\begin{align}
\frac{\|\Phi^{n+1}-\Phi^n\|^2_h}{\tau}&=\frac{1}{2}\left[E_h(\Phi^n)-E_h(\Phi^{n+1})\right]+\mu^{n+1/2}\frac{\|\Phi^{n+1}\|_h^2-\|\Phi^n\|_h^2}{2}\nonumber\\
&=\frac{1}{2}\left[E_h(\Phi^n)-E_h(\Phi^{n+1})\right],\nonumber
\end{align}
where the last equality holds true because of the fact $\|\Phi^{n+1}\|_h^2=\|\Phi^n\|_h^2$. 
It follows that $E_h(\Phi^n)\ge E_h(\Phi^{n+1})$. \qed
\end{proof}

Lemma \ref{lem:CNGF} indicates that, theoretically speaking, the CNGF-FD scheme \eqref{CNGF-FD} conserves mass and diminishes energy no matter the choice of time step $\tau$, and is thus unconditionally stable. 
However, the performance of the method relies on the proper choice of a nonlinear solver. 
Without a proper nonlinear solver, we may  lose the good properties in Lemma \ref{lem:CNGF}. 
For example, \eqref{CNGF-FD} can be solved using the Gauss-Seidel method as  follows.
\be\label{CNGF-FD-GS}
\frac{\phi^{n+1,m+1}_j-\phi^n_j}{\tau}=\left(\frac{1}{2}\delta_x^2-V_j-\beta\rho^{*,m}_j\right)\phi^{*,m+1}_j+(\delta\delta_x^2\rho^{*,m}_j+\tilde\mu^{m})\phi^{*,m}_j,
\ee
where $\phi^{n+1,0}_j:=\phi^n_j$, $\phi^{*,m}_j:=(\phi^n_j+\phi^{n+1,m}_j)/2$,  $\rho^{n+1,m}_j:=|\phi^{n+1,m}_j|^2$, $\rho^{*,m}_j:=(\rho^n_j+\rho^{n+1,m}_j)/2$  and  $\tilde\mu^{m}:=E_h(\Phi^{*,m})/\|\Phi^{*,m}\|_h^2$. 
And then $\Phi^{n+1}$ is  computed as
\be
\Phi^{n+1}:=\lim_{m\to\infty}\Phi^{n+1,m}.
\ee 
Equation \eqref{CNGF-FD-GS} implies that a linear system needs to be solved to get $\Phi^{n+1,m+1}$ and it can be shown that the solution always exists and is unique. However, it is not guaranteed that $\Phi^{n+1,m}$ has a limit as $m\to\infty$, which indicates that there should be some  restrictions on time step $\tau$ and mesh size $h$. 
Such restrictions will be shown in Fig. \ref{fig:stability} in Section \ref{sec:numeric} later. 
It is worth noticing that the restrictions are because of the specific nonlinear solver chosen rather than the  CNGF-FD scheme.

\subsection{Backward Euler finite difference discretization (BEFD)}\label{subsec:BEFD}
In this section, we will introduce the discrete normalized gradient flow and its detailed discretizations, especially the one with backward Euler in time and finite difference in space.

Despite the good properties shown in Lemma \ref{lem:CNGF}, the CNGF-FD scheme \eqref{CNGF-FD} is usually not the best choice  in practical computation since a complicated nonlinear equation needs to be solved for each time step, which might be time consuming and even cause severe stability issues. 
Modifications are needed to get schemes which are both efficient and stable. 
One way to due with the problem is to separate the steepest descent step and the normalization step in \eqref{CNGF1}, which gives us the discrete normalized gradient flow as follows. 
\begin{align}
&\phi_t=\frac{1}{2}\Delta\phi-V(\bx)\phi-\beta|\phi|^2\phi+\delta\Delta(|\phi|^2)\phi, \quad\bx\in\Omega, \,t_n<t<t_{n+1},\label{GF}\\
&\phi(\bx,t_{n+1})=\frac{\phi(\bx,t_{n+1}^+)}{\|\phi(\bx,t_{n+1}^+\|},\label{CNFD-normalization}\\
&\phi(\bx,t)|_{\partial \Omega}=0, \quad \phi(\bx,0)=\phi_0(\bx) \text{ with } \|\phi_0\|=1.
\end{align}
To update from $t=t_n$ to $t_{n+1}$, we first solve \eqref{GF} to get $\phi(\bx,t_{n+1}^+)$ and then get $\phi(\bx,t_{n+1})$ by the normalization step \eqref{CNFD-normalization}.

\begin{remark}
As shown in \cite{Wz1}, the solution of \eqref{GF}-\eqref{CNFD-normalization} may not preserve the energy diminishing property and the limiting solution might differ from the ground state. 
However, such difference is usually not obvious  and can be diminished by choosing a smaller time step. 
\end{remark}

Discretizing \eqref{GF}  via Crank-Nicolson (C-N)  in time and finite-difference (FD)  in space, we will get the CNFD scheme.
However, such scheme is far from satisfactory. The CNFD scheme is still nonlinear and, what makes things worse, there are high restrictions on time step and mesh size. 
As shown in Theorem 3.1 in \cite{Wz1}, even for the linear case where $\beta=\delta=0$,  we need to take $\tau=\mathcal{O}(h^2)$ to diminish energy. 
 
 A better choice would be the BEFD scheme, where \eqref{GF} is discretized via backward Euler in time and the nonlinear terms on the right-hand side are treated semi-implicitly. 
The idea comes from the success of the scheme for GPE where $\delta=0$ and $\beta>0$ \cite{Wz1}.
For the GPE case with $\beta>0$,  the nonlinear term $\beta|\phi|^2\phi$ is treated as $\beta|\phi^{n}|^2\phi^{n+1}$ when updating from $t=t_n$ to $t_{n+1}$, and the linear terms are treated implicitly. 
It is shown that the BEFD scheme is energy diminishing for any $\tau>0$ \cite{Wz1} when $V(\bx)\ge0$ and $\beta=\delta=0$, which implies a great advantage of the BEFD scheme over the CNFD scheme. 
When $\beta<0$, we need to approximate $\beta|\phi|^2\phi$ fully explicitly as $\beta|\phi^{n}|^2\phi^{n}$ to avoid possible instability issues. 

Following a similar idea,  we can construct the BEFD scheme for  MGPE.  
We adopt the same notations as in Section \ref{subsec:CNGF} and, for simplicity, only the 1D case where $\beta\ge0$ is considered.   
There are several ways to treat the $\delta$-nonlinear term semi-implicitly. 
Below is one possibility.   
For all $j=1,2,\dots,N-1$ and $n\ge0$, we have the BEFD scheme constructed as 
\begin{align}
&\frac{\tphi^{n+1}_j-\phi^n_j}{\tau}=\frac{1}{2}\delta_x^2\tphi^{n+1}_j-V_j\tphi^{n+1}_j-\beta|\phi^n_j|^2\tphi^{n+1}_j+\delta\delta_x^2(|\phi^n_j|^2)\tphi^{n+1}_j, \label{BEFD1}\\
&\Phi^{n+1}=\frac{\tPhi^{n+1}}{\|\tPhi^{n+1}\|_h}, \qquad \phi^{n+1}_0=\phi^{n+1}_N=0, \qquad \phi^0_j=\phi_0(x_j)\label{BEFD1_end}.
\end{align}
It is worth noticing that, with the BEFD scheme \eqref{BEFD1}-\eqref{BEFD1_end}, a linear equation of form $A^{(n)}\tPhi^{n+1}=F(\Phi^n)$, instead of a nonlinear one, needs to be solved for each time step. Easy to check that  $F(\Phi^n)=\Phi^n/\tau$ and 
\be\label{def:A}
A^{(n)}=I/\tau-D+V+\beta|\Phi^n|^2-\delta \delta_x^2(|\Phi^n|^2)
\ee
with $D:=(d_{jk})_{(N-1)\times(N-1)}$ defined as the finite difference approximation of $-\partial_{xx}/2$ where  
\be\label{FD-Delta}
d_{jk}:=\frac{1}{2h^2}
\begin{cases}
2,\quad j=k,\\
-1,\quad |j-k|=1,\\
0,\quad \text{otherwise}.
\end{cases}
\ee
and 
$V, \beta|\Phi^n|^2, \delta \delta_x^2(|\Phi^n|^2)$ are diagonal matrices whose diagonals are the corresponding vectors. 

The  BEFD scheme \eqref{BEFD1}-\eqref{BEFD1_end} is conditionally stable. In most cases, we need  $\tau\lesssim h^2$ as indicated in the following lemma. 
\begin{lemma}\label{lem:BEFD1}
When $\beta\ge0$ and $V\ge0$, one sufficient condition for the BEFD scheme \eqref{BEFD1}-\eqref{BEFD1_end} to be solvable is $\tau\lesssim h^2$. 
\end{lemma}
\begin{proof}
 One sufficient condition for the matrx $A^{(n)}$, with general choices of $V\ge0$ and $\beta\ge0$, to be nonsingular is
$\min\{1/\tau-\delta \delta_x^2(|\phi^n_j|^2)\}\ge0$, i.e. 
\be\label{BEFD1:cond}
\tau\le\frac{1}{\delta\delta_x^2(|\phi^n_j|^2)}, \text{ for } j=1,\dots,N-1.
\ee  
Noticing that $\delta_x^2 (|\Phi^n|^2)=-2D|\Phi^n|^2$, where the largest eigenvalue of $D$ is of order $\mathcal{O}(1/h^2)$ and $|\Phi^n|^2$ is bounded in $l_{\infty}$-norm \cite{mgpe-th},  we get $\delta_x^2(|\phi^n_j|^2)\lesssim 1/h^2$,which implies the sufficient condition $\tau\lesssim h^2$. \qed
\end{proof}

For the BEFD scheme \eqref{BEFD1}-\eqref{BEFD1_end}, the restriction on time step $\tau$ is mainly because of the possible negative sign of $\delta_x^2 (|\Phi^n|^2)$. 
To avoid the restriction, one possible way is to approximate the term $\delta\Delta(|\phi|^2)\phi$ fully explicitly. 
In this way, we get a different BEFD scheme as follows.
\begin{align}
&\frac{\tphi^{n+1}_{j}-\phi^n_j}{\tau}=\frac{1}{2}\delta_x^2\tphi^{n+1}_{j}-V(x_j)\tphi^{n+1}_{j}-\beta|\phi^n_j|^2\tphi^{n+1}_{j}+\delta\delta_x^2(|\phi^n_j|^2)\phi_j^n, \label{BEFD2}\\
&\Phi^{n+1}=\frac{\tPhi^{n+1}}{\|\tPhi^{n+1}\|_h}, \qquad \phi^{n+1}_0=\phi^{n+1}_N=0, \qquad \phi^0_j=\phi_0(x_j),\label{BEFD2_end}
\end{align}
where $j=1,2,\cdots,N-1$ and $n\ge0$.

The equation \eqref{BEFD2} can be reformulated in the matrix form $\tA^{(n)}\tPhi^{n+1}=\tF(\Phi^n)$, where $\tF(\Phi^n)=\Phi^n/\tau+\delta \delta_x^2(|\Phi^n|^2)\Phi^n$ and 
\be\label{def:tA}
\tA^{(n)}=I/\tau-D+V+\beta|\Phi^n|^2.
\ee
It is worth noting that the matrix $\tA^{(n)}$ is an M-matrix, which guarantees the solvability in each step. 
However, there are still restrictions on time steps for $\Phi^n$ to converge to the correct ground state $\Phi_g$. 
When $\beta=0$, the following lemma shows that the scheme \eqref{BEFD2} may be not  energy diminishing if we don't take time step satisfying $\tau\lesssim h^2$. 
\begin{lemma}\label{lem:BEFD2}
When $\beta=0$ and $V\ge0$, one sufficient condition for 
$E_h(\tPhi^{n+1})\le E_h(\Phi^n)$, for all $n\ge0$,  is $\tau\lesssim h^2$. 
\end{lemma}
\begin{proof}
For simplicity, denote $\delta_t\Phi^n=\tPhi^{n+1}-\Phi^n$ and $\|\Phi\|_V=(V\Phi,\Phi)$, where $(\cdot,\cdot)$ denotes the inner product of two vectors. 
Multiplying both sides of \eqref{BEFD2} by $\tphi^{n+1}_{j}-\phi^n_j$ and summing over $j$, we get 
\begin{align}
&\frac{2\|\delta_t\Phi^n\|^2}{\tau}=(\delta_x^2\tPhi^{n+1},\delta_t\Phi^n)-2(V\tPhi^{n+1},\delta_t\Phi^n)+2\delta(\delta_x^2(|\Phi^n|^2)\Phi^n,\delta_t\Phi^n)\nonumber\\
&=\frac{1}{2}(\|\delta_x\Phi^n\|^2-\|\delta_x\tPhi^{n+1}\|^2-\|\delta_x\delta_t\Phi^{n}\|^2)+(\|\Phi^{n}\|_V-\|\tPhi^{n+1}\|_V-\|\delta_t\Phi^{n}\|_V)\nonumber\\
&+\frac{1}{2}\left[\|\nabla(|\Phi^n|^2)\|^2-\|\nabla(|\tPhi^{n+1}|^2)\|^2+\|\delta_x\delta_t(|\Phi^n|^2)\|^2\right]-\left(\delta_x^2(|\Phi^n|^2),(\delta_t\Phi^n)^2\right)\nonumber\\
&\le E_h(\Phi^n)-E_h(\tPhi^{n+1})+\frac{1}{2}\|\delta_x\delta_t(|\Phi^n|^2)\|^2-\left(\delta_x^2(|\Phi^n|^2),(\delta_t\Phi^n)^2\right).\nonumber
\end{align}
Obviously, one sufficient condition for the energy diminishing is 
\be\label{proof:BEFD2_1}
\frac{2\|\delta_t\Phi^n\|^2}{\tau}\ge\frac{1}{2}\|\delta_x\delta_t(|\Phi^n|^2)\|^2-\left(\delta_x^2(|\Phi^n|^2),(\delta_t\Phi^n)^2\right).
\ee
The boundedness of $\Phi^n$ for arbitrary $n$ \cite{mgpe-th} indicates the boundedness of $|\Phi^n|^2$. 
Furthermore, we have $\|\delta_t|\Phi^n|^2\|_{\infty}\lesssim\|\delta_t\Phi^n\|_{\infty}$. 
Following a similar argument as in the proof of Lemma \ref{lem:BEFD1}, we get 
$\|\delta_x^2(|\Phi^n|^2)\|_{\infty}\lesssim1/h^2$.
As a result, 
\be\label{proof:BEFD2_2}
\left|\left(\delta_x^2(|\Phi^n|^2),(\delta_t\Phi^n)^2\right)\right|\le\|\delta_x^2(|\Phi^n|^2)\|_{\infty}\left(1,(\delta_t\Phi^n)^2\right)\lesssim\|\delta_t\Phi^n\|^2/h^2.
\ee
 Besides, the Bramble-Hilbert lemma and a standard scaling argument gives
 \be\label{proof:BEFD2_3}
 \|\delta_x\delta_t(|\Phi^n|^2)\|^2\lesssim \frac{ \|\delta_t(|\Phi^n|^2)\|^2}{h^2}.
 \ee
 Combing \eqref{proof:BEFD2_1}, \eqref{proof:BEFD2_2} and \eqref{proof:BEFD2_3}, we get the sufficient condition $\tau\lesssim h^2$. \qed
\end{proof}

Lemma \ref{lem:BEFD1} and  \ref{lem:BEFD2} indicate the possible strict restrictions on time step $\tau$ no matter whether the term $\delta\Delta(|\phi|^2)\phi$ is treated semi-implicitly or fully explicitly. 
Later in Section \ref{sec:numeric}, we will show numerically that such restrictions are necessary in many cases for the solution to converge to the desired ground state. 
The high restrictions on time step imply the BEFD schemes \eqref{BEFD1}-\eqref{BEFD1_end} and    \eqref{BEFD2}-\eqref{BEFD2_end} proposed in this section are not satisfactory. 
The main reason is that we didn't treat the nonlinear term  $\delta\Delta(|\phi|^2)\phi$ properly. 
In the next section, we will introduce a new scheme, named the backward Euler finite difference scheme with splitting (BEFD-splitting), which splits the term  $\delta\Delta(|\phi|^2)\phi$ into two parts with each part dealt with separately. 
It will be shown numerically that the new method will significantly improve the scheme in the sense that a much larger time step can be adopted for the scheme which is almost unrelated to the mesh size we choose.

\section{Backward Euler discretization with attractive-repulsive splitting}\label{sec:splitting}

In this section, we introduce a new BEFD scheme which can be numerically proven to be much more efficient and stable than the BEFD schemes \eqref{BEFD1}-\eqref{BEFD1_end} and \eqref{BEFD2}-\eqref{BEFD2_end} proposed in the last section. 

\subsection{A new gradient flow with attractive-repulsive splitting}
\label{subsec:gradientflow_splitting}

In this section, we will introduce a new gradient flow with its semi-discretization, which will be used to construct a new BEFD scheme suitable for practical computation.  
The ideal scheme should satisfy the following properties. 
\begin{itemize}
\item
The update for each step should be simple, which is necessary for the scheme to be efficient.  
\item
The scheme should have good stability properties, i.e. the time step shouldn't be highly restricted by the mesh size.
 In this way, we can improve the accuracy of the solution without much extra computational cost.  \end{itemize}
The CNGF-FD scheme \eqref{CNGF-FD} do not satisfy the first requirement since a complicated nonlinear system needs to be solved for each step.  
The BEFD schemes \eqref{BEFD1}-\eqref{BEFD1_end} and \eqref{BEFD2}-\eqref{BEFD2_end}  do not satisfy the second condition since an extremely small time step, which is of order $\mathcal{O}(h^2)$, is needed.  

The instability issue of the BEFD schemes \eqref{BEFD1}-\eqref{BEFD1_end} and \eqref{BEFD2}-\eqref{BEFD2_end} is due to the nonlinear term $\delta\Delta(|\phi|^2)\phi$ since the schemes work well when $\delta=0$ as shown in \cite{Wz1}.  
One possible reason underlying is that $\Delta(|\phi|^2)$ is much less smooth than $|\phi|^2$.  
A small error in $\phi$ would affect $\Delta(|\phi|^2)$ in a significant way, which indicates that instead of treating the part $\Delta(|\phi|^2)$ as  $\Delta(|\phi^n|^2)$ as a whole, 
we should split it in a proper way and deal with each part separately. 
 A simple computation gives that 
\be\label{eq:split}
\delta\Delta(|\phi|^2)\phi=2\delta|\phi|^2\Delta\phi+2\delta|\nabla\phi|^2\phi.
\ee
Combining the first term on the right hand side of \eqref{eq:split}, i.e.  $2\delta|\phi|^2\Delta\phi$, with the linear term $\frac{1}{2}\Delta\phi$,  we get an equivalent new form of the original gradient flow \eqref{GF}. To be more specific, the new gradient flow reads as follows.
\begin{align}
\phi_t&=\frac{1}{2}\Delta\phi-V(\bx)\phi-\beta|\phi|^2\phi+2\delta|\phi|^2\Delta\phi+2\delta|\nabla\phi|^2\phi,\nonumber\\
&=(\frac{1}{2}+2\delta|\phi|^2)\Delta\phi-V\phi-\beta|\phi|^2\phi+2\delta|\nabla\phi|^2\phi,\quad\bx\in\Omega, \,t\ge0.\label{MGPE_equiv}
\end{align}

The equation \eqref{MGPE_equiv} is equivalent to the original gradient flow \eqref{GF} in the continuous level. 
However, it will lead to different schemes after discretization in time. 
The trick  is that we treat the term $2\delta|\phi|^2\Delta\phi$ semi-implicitly as $2\delta|\phi^n|^2\Delta\phi^{n+1}$ while treat the term $2\delta|\nabla\phi|^2\phi$ fully explicitly as $2\delta|\nabla\phi^n|^2\phi^n$. 
And then we get the following semi-discretized scheme from $t_n$ to $t_{n+1}$
\begin{align}\label{BE-splitting}
&\frac{\tphi^{n+1}-\phi^n}{\tau}=\left[\left(\frac{1}{2}+2\delta\rho^n\right)\Delta-V(\bx)-\beta\rho^n\right]\tphi^{n+1}+2\delta|\nabla\phi^n|^2\phi^n,\\
&\phi^{n+1}=\frac{\tphi^{n+1}}{\|\tphi^{n+1}\|}, \quad\text{ with }\tphi^{n+1}|_{\partial \Omega}= 0\text{ and } \rho^n=|\phi^n|^2, \label{BE-splitting-norm}
\end{align}
where $\bx\in\Omega$, $n\ge0$ and $\beta\ge0$. If $\beta<0$, we need to change $\beta\rho^n\tphi^{n+1}$ to $\beta\rho^n\phi^{n}$.

It is obvious that the new scheme \eqref{BE-splitting} is linear in $\tphi^{n+1}$. Furthermore, the coefficient before $\Delta$, i.e. $\frac{1}{2}+2\delta\rho^n$, is always positive, which guarantees the existence and uniqueness of $\tphi^{n+1}$ for any $n$. 
Besides, by splitting the term $\delta\Delta(|\phi|^2)\phi$ into two parts, we avoid the explicit treatment of $\Delta(|\phi|^2)$ as a whole, which indicates the possibility to get rid of the strict restrictions on $\tau$. All these features imply that the new gradient flow \eqref{BE-splitting} could be used to construct new numerical schemes suitable for practical computation. 

The idea of treating the two terms on the right-hand side of  \eqref{eq:split} separately comes from the convex-concave splitting proposed in, for example, \cite{ShenWWW,WuZZ,Yuille} and so on,
where the term corresponding to a convex energy is treated implicitly while the term corresponding to a concave term is treated explicitly.  
It is proved in \cite{Yuille} that, without discrete normalization, the gradient flow will always be energy diminishing if the convex-concave splitting is applied.
Multiplying both terms in \eqref{eq:split} by $\phi$, integrating over  domain $\Omega$ and doing integration by part, we get  the corresponding energies of the terms $2\delta|\phi|^2\Delta\phi$ and $2\delta|\nabla\phi|^2\phi$ as 
 \be
 H_1=\frac{3\delta}{2}\int_{\Omega}|\nabla(|\phi|^2)|^2\,d\bx\text{ and }
 H_2=-\frac{\delta}{2}\int_{\Omega}|\nabla(|\phi|^2)|^2\,d\bx,
 \ee 
 respectively. 
 Obviously, $H_1$ is convex in $\phi$ while $H_2$ is concave, which indicates the proper discretization of  \eqref{MGPE_equiv} is  \eqref{BE-splitting}. 
 Slightly different from the convex-concave splitting, the term $2\delta|\phi|^2\Delta\phi$ is treated semi-implicitly instead of fully implicitly for the computational efficiency. 

The scheme \eqref{BE-splitting} can be understood from a physics point of view as well by checking the corresponding energy terms.  
 In physics, a positive energy term corresponds to a repulsive energy, which will stablize the system and should be discretized implicitly, while a negative energy term corresponds to an attractive energy, which will, on the contrary, destabilize the system and must be discretized in a fully explicit way. 
Noticing that the energies corresponding to $2\delta|\phi|^2\Delta\phi$ and $2\delta|\nabla\phi|^2\phi$ are $H_1$ and $H_2$, respectively, and $H_1\ge0$  and  $H_2\le0$ hold  true for any function $\phi$, $2\delta|\phi|^2\Delta\phi$ should be treated implicitly while  $2\delta|\nabla\phi|^2\phi$ must be treated explicitly. 
 The signs of $H_1$ and $H_2$ matter here. 
This idea also explains why we need to change  $\beta|\phi^n|^2\tphi^{n+1}$  to $\beta|\phi^n|^2\phi^{n}$ if $\beta<0$.  
From this point of view, the splitting is  based on the repulsive-attractive splitting of the energy.
As a result, we name our new method to be the \textbf{normalized gradient flow with  repulsive-attractive splitting}.

\subsection{Backward Euler finite difference discretization with splitting}\label{subsec:BEFD-splitting}
In this section, we construct the \textbf{backward Euler finite difference scheme with splitting (BEFD-splitting)}
by discretizing \eqref{BE-splitting} in space via the finite difference.
For simplicity, we adopt the notations as in Section \ref{subsec:CNGF} and consider the 1D problem.  
When $\beta\ge0$,  we have
\begin{align}\label{BEFD-splitting}
&\frac{\tphi^{n+1}_j-\phi^n_j}{\tau}=\left[\left(\frac{1}{2}+2\delta|\phi^n_j|^2\right)\delta_x^2-V_j-\beta|\phi^n_j|^2\right]\tphi^{n+1}_j+2\delta|\delta_x^+\phi^n_j|^2\phi^n_j,\\
&\Phi^{n+1}=\frac{\tPhi^{n+1}}{\|\tPhi^{n+1}\|_h}, \qquad \phi^{n+1}_0=\phi^{n+1}_N=0, \qquad \phi^0_j=\phi_0(x_j).\label{BEFD-splitting-norm}
\end{align}
The equation \eqref{BEFD-splitting}  can be written in matrix form as  
\be\label{BEFD-matrix}
A^{(n)}\tPhi^{n+1}=F(\Phi^n)
\ee
where $A^{(n)}=I+\tau\left[\left(I+4\delta|\Phi^n|^2\right)D+V+\beta|\Phi^n|^2\right]$, $D$ is defined in \eqref{FD-Delta} and
$
F(\Phi^n)=(I+2\delta\tau|\delta_x^+\Phi^n|^2)\Phi^n.
$

Noticing that the matrix $A^{(n)}$ is  tridiagonal and diagonal dominant, the equation \eqref{BEFD-matrix} can be solved efficiently by the  tridiagonal matrix algorithm (TDMA) with the computational cost $\mathcal{O}(N)$. 
The scheme \eqref{BEFD-splitting} can be generalized to multidimensional problems with tensor product grids in a straightforward way. 
In such cases, the matrix $A^{(n)}$ is no longer tridiagonal and TDMA cannot be applied.  
However, it is a \textbf{sparse M-matrix}, which indicates that we can use iterative methods, such as  Gauss-Seidel method,  to get $\tPhi^{n+1}$ efficiently.

It is remarkable that the BEFD-splitting scheme \eqref{BEFD-splitting}-\eqref{BEFD-splitting-norm} is not unconditionally stable. The stability region is difficult to be theoretically analysed. However, as shown by the numerical experiments in Section \ref{subsec:stability}, the scheme is much less restricted than the BEFD schemes \eqref{BEFD1}-\eqref{BEFD1_end} and  \eqref{BEFD2}-\eqref{BEFD2_end}. 
It can be concluded that the new scheme \eqref{BEFD-splitting}-\eqref{BEFD-splitting-norm} satisfies the properties proposed in Section \ref{subsec:gradientflow_splitting} and  has a great advantage over all our previous methods. 

\subsection{Backward Euler pseudo-spectral discretization with splitting}\label{subsec:BESP}
In this section, we construct the \textbf{backward Euler pseudo-spectral scheme with splitting  (BESP-splitting)} 
by discretizing \eqref{BE-splitting} in space via Fourier spectral method. 
Compared to the finite difference method, the spectral method has the advantage of high accuracy for regular domains and smooth solutions. 

For simplicity, 
consider the problem in 1D defined in $\Omega=(a,b)$ and use the same notations as in Section  \ref{subsec:CNGF}. 
Introduce  $\mu_l=\frac{2\pi l}{b-a}$ and 
\be
\hat\Phi:=(\hphi_{-N/2},\hphi_{-N/2+1},\dots,\hphi_{N/2-1})
\ee
 to be the discrete Fourier transform of $\Phi=(\phi_0, \phi_1, \dots , \phi_{N-1})$. It is easy to check that $\hphi_{l}=\hphi_{l+N}$. 
Then \eqref{BEFD-splitting} can discretized as 
\begin{align}\label{BESP-splitting}
&\frac{\tphi^{*}_j-\phi^n_j}{\tau}=\left[\left(\frac{1}{2}+2\delta|\phi^n_j|^2\right)D_{xx}^s-V_j-\beta|\phi^n_j|^2\right]\tphi^{*}_j+2\delta(D_x^s\phi^n_j)^2\phi^n_j,\\
&\Phi^{n+1}=\frac{\tPhi^{*}}{\|\tPhi^{*}\|_h}, \quad\text{ with }\quad \tphi^{*}_0=\tphi^{*}_N=0, \qquad \phi^0_j=\phi_0(x_j).\label{BESP-splitting-norm}
\end{align}
where $D_{xx}^s$ and $D_x^s$ are the pseudo-spectral differential operators approximating $\partial_{xx}$ and $\partial_x$, respectively,  and defined as
\begin{equation}
D_{x}^s \phi_j=\sum_{l=-N/2}^{N/2-1} i\mu_l\hat\phi_l e^{i\mu_l(x_j-a)},
\quad
D_{xx}^s \phi_j=-\sum_{l=-N/2}^{N/2-1} \mu_l^2\hat\phi_l e^{i\mu_l(x_j-a)}.
\end{equation}
To preserve the high accuracy of the spectral method, the discretized energy needs to be computed in the following way \cite{Bao2013}, 
\be\label{energy:sp}
E_h^{SP}(\Phi):=\frac{b-a}{2}\sum_{l=-N/2}^{N/2-1}\left[\mu_l^2\hphi_l^2+\delta \mu_l^2\hrho_l^2\right]+h\sum_{j=1}^{N-1}\left[V|\phi_j|^2+\frac{\beta}{2}|\phi_j|^4\right],
\ee
where $\mu_l$, $\hat\phi_l$ are defined as before and  $\hat\rho:=(\hrho_{-N/2},\hrho_{-N/2+1},\dots,\hrho_{N/2-1})$ is  the discrete Fourier transform of $\rho:=|\Phi|^2$.

The equation \eqref{BESP-splitting} can be written in a matrix form as well. 
Define matrices  
$
W=(w_{jk})\in\mathbb{R}^{N\times N}  \text{ with } w_{jk}=e^{i2\pi jk}$ 
 and 
$
\Xi=\text{diag}(\mu_{k})\in\mathbb{R}^{N\times N}
$
for $ j=0,1,\dots,N-1$ and $k=-N/2, -N/2+1, \dots, N/2-1$. Denote 
\be
D^{\rm{(sp)}}_{x}=iW\Xi W^{-1}\qquad D_{xx}^{\rm{(sp)}}=-W\Xi^2 W^{-1}. 
\ee
Then the equation \eqref{BESP-splitting} can be written as 
\be
A^{(n)}_{\rm{sp}}\tPhi^*=F(\Phi^n), 
\ee 
where $A^{(n)}_{\rm{sp}}=I+\tau\left[-\left(I+4\delta|\Phi^n|^2\right)D^{\rm{(sp)}}_{xx}/2+V+\beta|\Phi^n|^2\right]$ and
$
F(\Phi^n)=(I+2\delta\tau|D^{\rm{(sp)}}_{x}\Phi^n|^2)\Phi^n.
$

Since the fast Fourier transform (FFT) cannot be applied to evaluate $\left[\left(\frac{1}{2}+2\delta|\Phi^n|^2\right)D_{xx}^s\right]^{-1}$, 
the equation \eqref{BESP-splitting} cannot be solved in a direct way or iteratively by the fixed point method as in \cite{Bao2013,Wz1} with the computational cost $\mathcal{O}(N\log N)$. 
One alternative way is through a Krylov subspace method, such as BiCGSTAB \cite{GPELab,Krylov}. 
As shown in \cite{GPELab,Krylov}, the  main computational cost of BiCGSTAB is the computation of the matrix-vector multiplication $A^{(n)}_{\rm{sp}} v$ for some vector $v$. 
Although the matrix $A^{(n)}_{\rm{sp}}$ is full and asymmetric, the matrix-vector multiplication can be effectively evaluated via FFT with the computational cost $\mathcal{O}(N\log N)$, which guarantees that the total computational cost per step is $\mathcal{O}(N\log N)$ as well.
Compared to the BEFD-splitting scheme \eqref{BEFD-splitting}-\eqref{BEFD-splitting-norm}, the BESP-splitting scheme \eqref{BESP-splitting}-\eqref{BESP-splitting-norm} has the advantage of high accuracy while
the order of the computional cost per step is almost the same,  
which implies the scheme is usually more effective in practice.

\section{Numerical results}\label{sec:numeric}

In this section, we will show numerical experiments with the BEFD-splitting scheme \eqref{BEFD-splitting}-\eqref{BEFD-splitting-norm}  and the BESP-splitting scheme \eqref{BESP-splitting}-\eqref{BESP-splitting-norm}. 
To show the great superiority of our new schemes, 
we will compare the BEFD-splitting/BESP-splitting scheme with other schemes, including CNGF-FD  \eqref{CNGF-FD}, the naive BEFD schemes \eqref{BEFD1}-\eqref{BEFD1_end} and \eqref{BEFD2}-\eqref{BEFD2_end}, and also the extension of the regularized Newton method proposed in \cite{Ruan_thesis,BaoWuWen}. 
Then the spatial accuracy of the BEFD-splitting scheme \eqref{BEFD-splitting}-\eqref{BEFD-splitting-norm} and the BESP-splitting scheme \eqref{BESP-splitting}-\eqref{BESP-splitting-norm} will be tested numerically. 
Finally, we will apply our new methods to compute 
ground states of the MGPE defined in multidimensional space with general external potentials as well as 
to compute the first excited states. 

\subsection{Energy decay and convergence test}\label{subsec:stability}

In this section, we show the conditional energy decay property of  the BEFD-splitting scheme \eqref{BEFD-splitting}-\eqref{BEFD-splitting-norm}  and the BSEP-splitting scheme \eqref{BESP-splitting}-\eqref{BESP-splitting-norm}.

\begin{example}\label{exmp1}
 Consider the MGPE in 1D under the harmonic potential $V(x)=x^2/2$.
The initial data is chosen to be 
\be\label{initial}
\phi_0(x)=\frac{e^{-\frac{x^2}{2}}}{\pi^{1/4}}.
\ee
Two choices of the parameters are considered: 
(I) $\beta=\delta=10$, 
(II) $\beta=-10, \,\delta=10$.
\end{example}


Fig. \ref{fig:energy_dynamic} shows the energy evolution via the BEFD-splitting scheme \eqref{BEFD-splitting}-\eqref{BEFD-splitting-norm} or the BESP-splitting scheme \eqref{BESP-splitting}-\eqref{BESP-splitting-norm} for Case I in Example \ref{exmp1}.  
As shown in the figure, the solution will converge to a steady state  with a small time step $\tau=0.01$.  
As time step increases, the energy will begin to oscillate in time. 
When the time step is large enough, say $\tau=2$, the energy will diverge and the solutions no longer converge to the ground state.  

Different from the classical definition of instability, we say the scheme is \textbf{unstable} with the given parameters if we can't get the correct ground state.
In this sense, the numerical test implies that the BEFD-splitting scheme and  BESP-splitting scheme are only conditionally stable. 
Fig. \ref{fig:BESP_dtdh} shows the effect of the time step and mesh size on the stability region of BESP-splitting.  
As shown in the figure, the scheme is more stable with a small time step $\tau$ and a large mesh size $h$,
 and the stability of the scheme would not be much affected by $h$ if $h$ is small enough.  
In practical numerical computation, a smaller mesh size is needed sometimes for better accuracy and a relatively larger time step is preferred for a faster convergence. 
Therefore, we need to balance the stability, accuracy and efficiency by choosing a proper mesh size and time step in practical numerical computation.

The performance of the gradient flow method depends on the choice of the initial value as well. A good choice of the initial value would not only accelerate the convergence of the solutions, but also improve the stability of the scheme. 
For MGPE with weak or strong nonlinearity and with special external potentials, such as box potential and harmonic potential, the proper choices can be found in \cite{Ruan_thesis,mgpe-asym}. 
For MGPE with general external potentials, the multigrid technique (see \cite{Briggs,Wesseling} and references therein) using a hierarchy of discretizations would be a good choice to increase the stability and efficiency of the schemes. 
In Table \ref{tab:multigrid}, we show how the multigrid technique improves the stability of the scheme.  We consider the case with $h=0.25$ and $\tau=0.01$. 
To apply the multigrid technique, we start with the problem with $h=0.5$ and $\tau=0.01$ to get the ground state $\phi_g^c$ on the coarse grid. Then $\phi_g^c$ is refined to be on the grid with mesh size $h=0.25$ in a proper way. 
With the refined function as the initial data, we apply the BEFD-splitting/BESP-splitting method again to get the ground state. 
As shown in Table \ref{tab:multigrid}, the stability is significantly improved with this simple technique.

\begin{figure}[h!]
\centerline{\psfig{figure=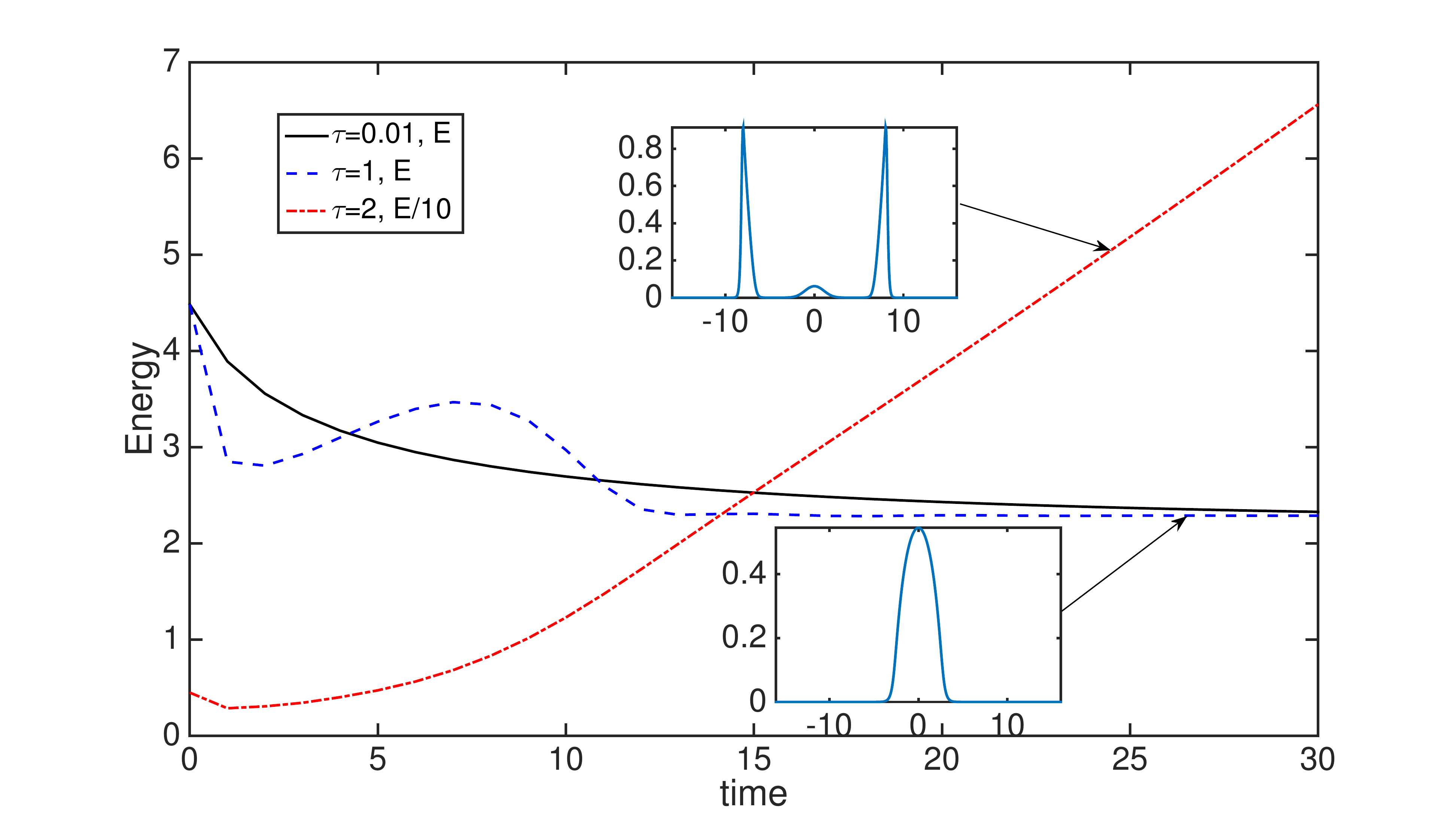,height=4.5cm,width=6.5cm,angle=0}
\psfig{figure=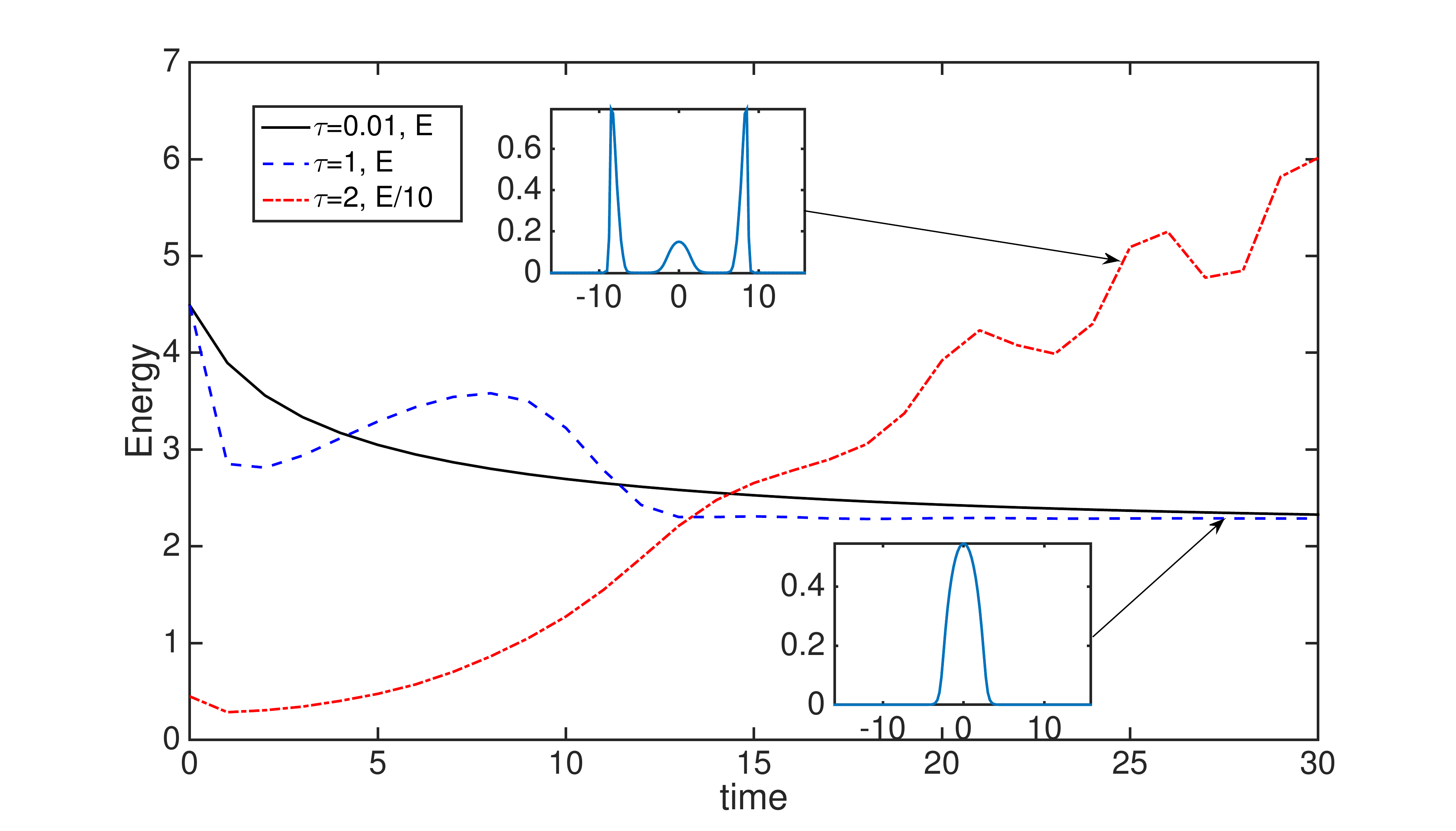,height=4.5cm,width=6.5cm,angle=0}}
\caption{Energy dynamics via  BEFD-splitting  \eqref{BEFD-splitting}-\eqref{BEFD-splitting-norm} (left) and  BESP-splitting  \eqref{BESP-splitting}-\eqref{BESP-splitting-norm} (right) with different choices of time steps.}
\label{fig:energy_dynamic}
\end{figure}

\begin{figure}[h!]
\centerline{\psfig{figure=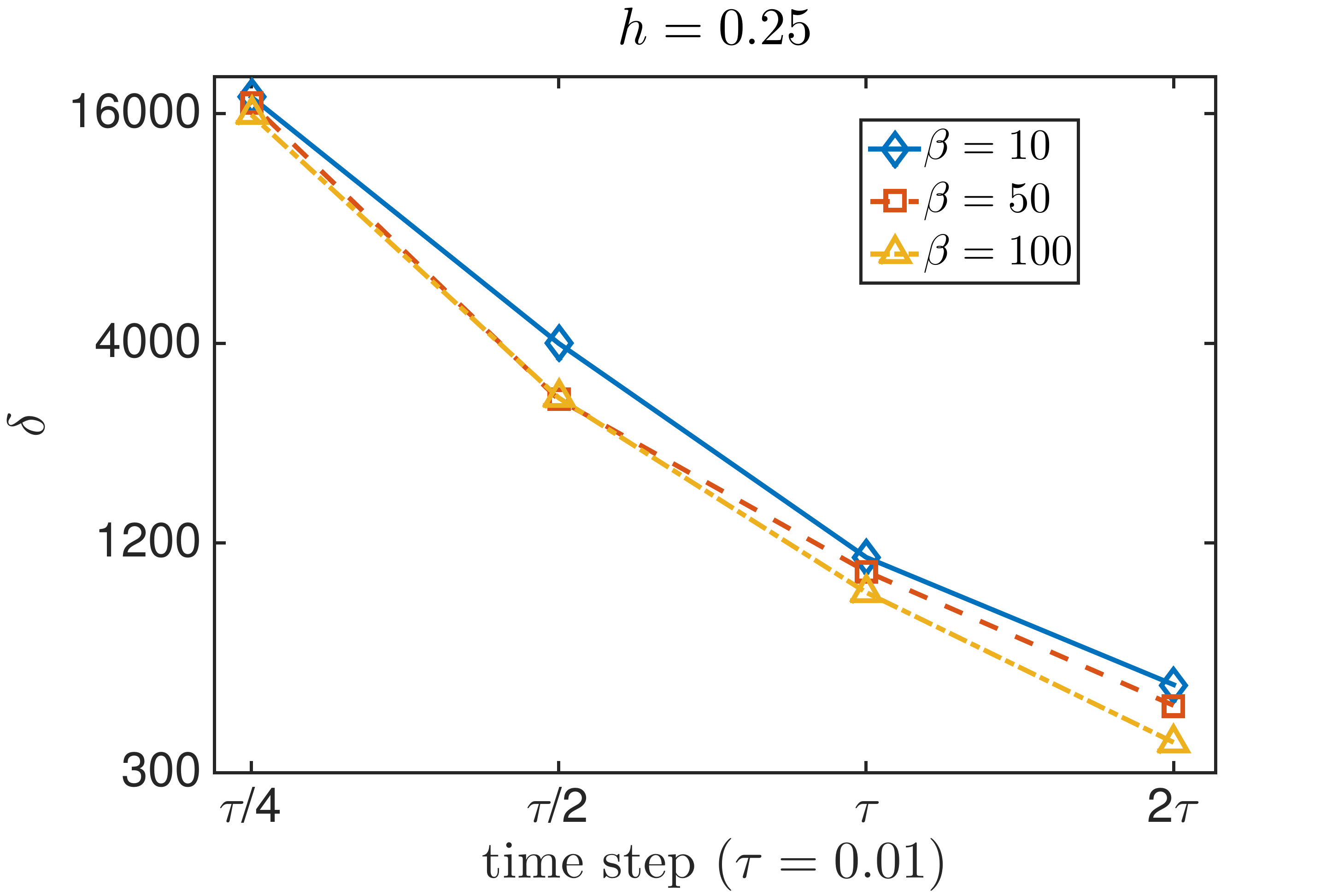,height=4.5cm,width=6.5cm,angle=0}
\psfig{figure=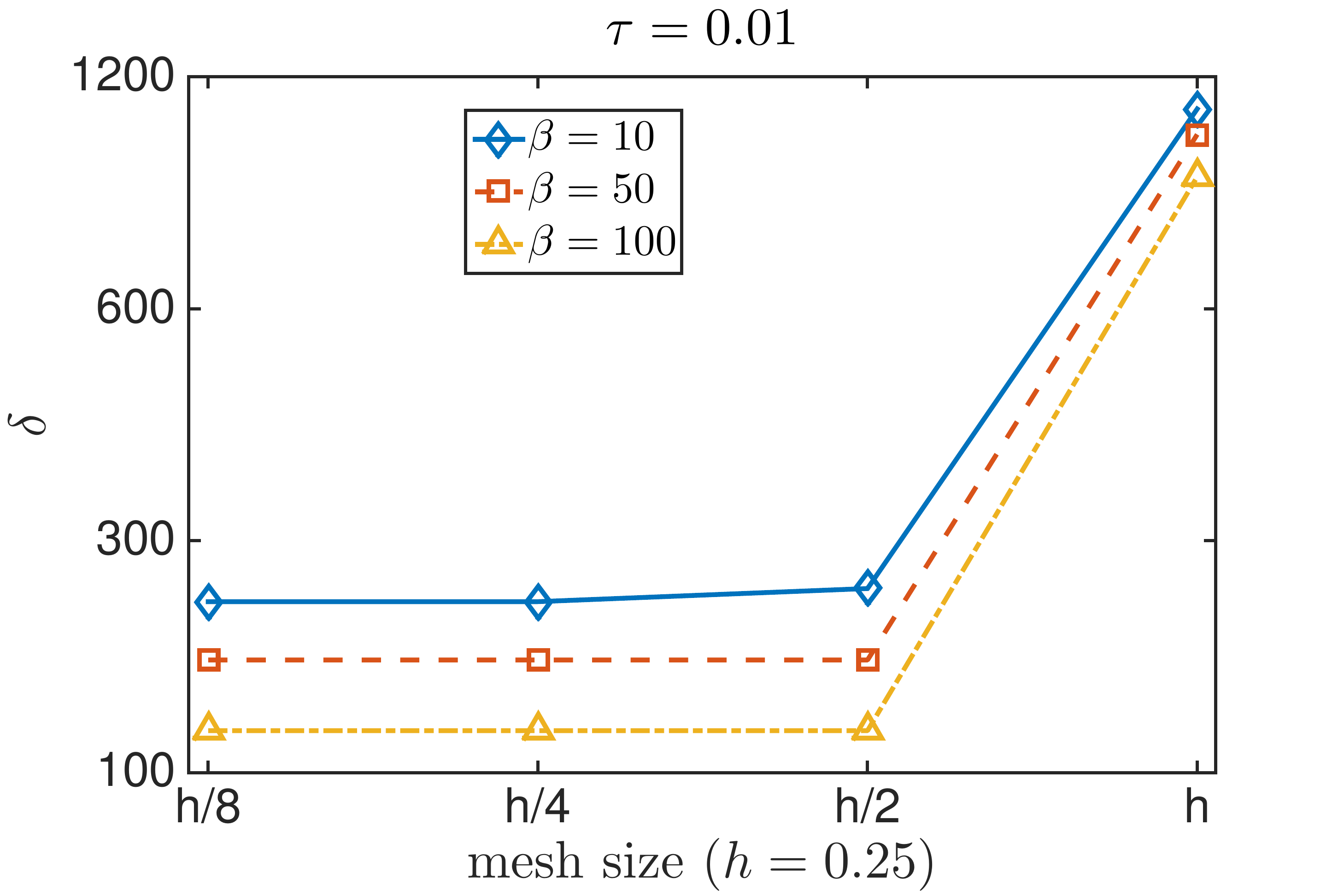,height=4.5cm,width=6.5cm,angle=0}}
\caption{Effect of time step $\tau$ and mesh size $h$ on the stability region of BESP-splitting \eqref{BESP-splitting}-\eqref{BESP-splitting-norm}. One remark for the figure on the right is that, when mesh size is $2h$, the scheme is stable with $\delta>50000$ for all three choices of $\beta$.}
\label{fig:BESP_dtdh}
\end{figure}

\begin{table}
  \centering
  \begin{tabular}{c c}
    \hline
    BEFD-splitting & $\delta$ \\
    \hline
    \hline
    without multigrid & $\approx$4000  \\
    \hline
    multigrid  &  $>$50000 \\
    \hline
  \end{tabular}
\quad
  \begin{tabular}{c c}
    \hline
    BESP-splitting & $\delta$ \\
    \hline
    \hline
    without multigrid & $\approx$890  \\
    \hline
    multigrid  & $>$50000 \\
    \hline
  \end{tabular}
  \caption{Largest possible $\delta$ computable by BEFD-splitting \eqref{BEFD-splitting}-\eqref{BEFD-splitting-norm} or BESP-splitting \eqref{BESP-splitting}-\eqref{BESP-splitting-norm} with $h=0.25$, $\tau=0.01$ and $\beta=100$. The scheme with the multigrid technique starts from $h=0.5$ and uses the refined solution as the initial data for the problem where $h=0.25$. }
  \label{tab:multigrid}
\end{table}

\subsection{Comparison with different schemes}
In this section, we compare the BEFD-splitting/BESP-splitting schemes with other schemes, especially the CNGF-FD  \eqref{CNGF-FD} and the BEFD \eqref{BEFD1}-\eqref{BEFD1_end} and \eqref{BEFD2}-\eqref{BEFD2_end} proposed in this paper, on the stability region. 
The results are summarized in Fig. \ref{fig:stability}. 

Fig. \ref{fig:stability_sub1} shows the stability region when $\beta=10$ and $\delta=10$. The method will be stable with the time step $\tau$ and mesh size $h$  if the point $(h,\tau)$ is below the line, and vice versa. Therefore, the best method will be the  least restrictive one, i.e. the borderline of the stability region should be on the top.  
As shown in the figure, when $\beta=\delta=10$, the time step can be chosen to be $\mathcal{O}(1)$ for the  BEFD-splitting/BESP-splitting scheme. However, for the other three methods, we need $\tau\approx\mathcal{O}(h^2)$
, which is very restrictive.  
In this sense, we can conclude that the BEFD-splitting method and the BESP-splitting method are much more stable than the other methods.

Fig. \ref{fig:stability_sub2} shows the borderline of the stability region when $\tau=0.01$ and $h=0.25$. The scheme is stable if the point $(\beta,\delta)$ is below the line. Similarly, a good scheme should be the one applicable to the case with strong nonlinearity, i.e. large $\delta$ and $\beta$. Therefore. the borderline for the best scheme should be on the top. 
As shown in the figure, the BEFD-splitting scheme is the best one and the BEFD-splitting/BESP-splitting schemes are much better than the other methods. 
The huge difference indicates that the attractive-repulsive splitting of the term $\delta\Delta(|\phi|^2)\phi$ does improve the stability significantly, which makes the gradient method suitable for practical numerical computation.

\begin{figure}[h!]
\begin{subfigure}{.5\textwidth}
  \centering
  \includegraphics[height=4.5cm,width=6.5cm]{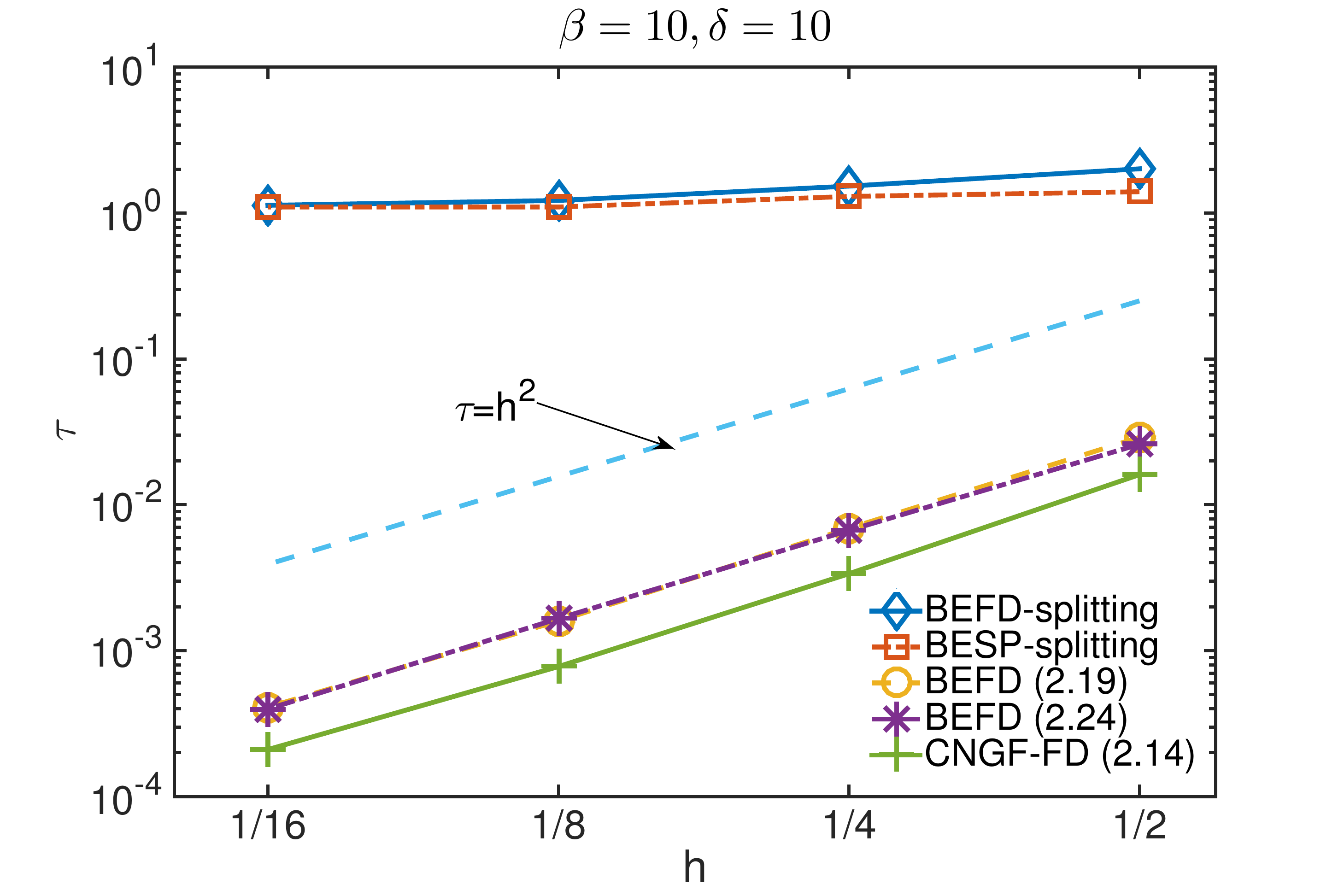}
  \caption{Fix $\beta$ and $\delta$.}
  \label{fig:stability_sub1}
\end{subfigure}%
\begin{subfigure}{.5\textwidth}
  \centering
  \includegraphics[height=4.5cm,width=6.5cm]{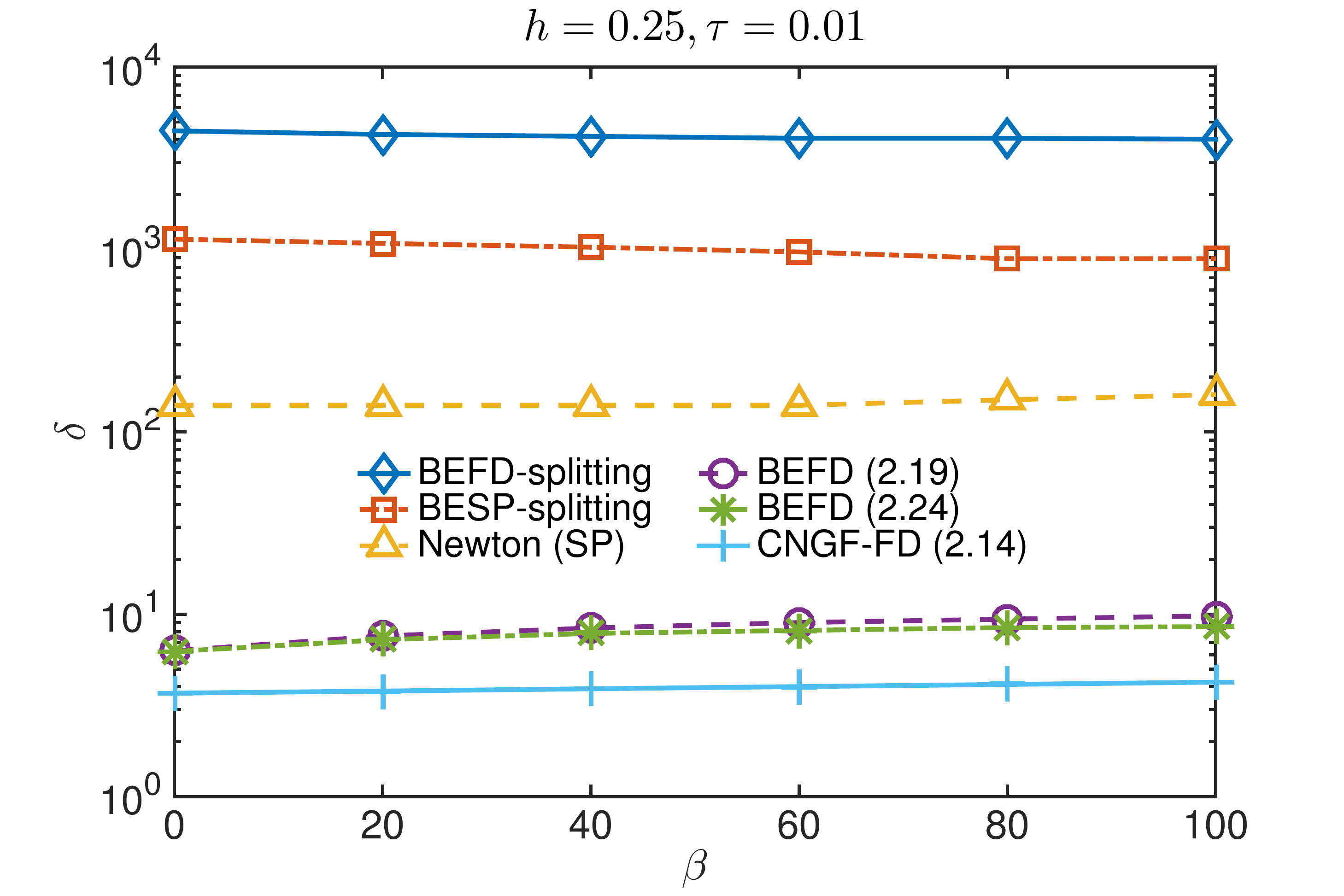}
  \caption{Fix $h$ and $\tau$.}
  \label{fig:stability_sub2}
\end{subfigure}
\caption{Stability test for different schemes. The lines denote the borderlines of the stability region and the part below the line corresponds to the region where the scheme is stable. }
\label{fig:stability}
\end{figure}

Finally, we compare the BEFD-splitting/BESP-splitting scheme with the regularized Newton method proposed in \cite{BaoWuWen}, which minimizes the discrete energy functional directly under the normalization constraints. The extension of the method to MGPE is straightforward \cite{Ruan_thesis}. 
We compare the stability region of the regularized Newton method discretized via pseudo-spectral method in space with the BEFD-splitting/BESP-splitting method and the result is shown in Fig. \ref{fig:stability_sub2}. 
Numerical experiments show that the regularized Newton method will give an oscillatory solution, which is apparently different from the exact solution, when $\delta$ is large. 
In such cases, we call the method unstable, which is consistent with our previous definition of the stability of a scheme. 
As shown in the figure, the regularized Newton method is much less stable than the BEFD-splitting/BESP-splitting method. 
One result not shown in Fig. \ref{fig:stability_sub2} is that the regularized Newton method with finite difference discretization in space is quite stable, even when $\delta$ is extremely large. 
However, as indicated in Table \ref{tab:multigrid}, a simple multigrid technique will improve the stability of the BEFD-splitting/BESP-splitting scheme signicantly and make the schemes competitive with the regularized Newton method with finite difference discretization in space.  

\subsection{Spatial accuracy}
In this section, we test the spatial accuracy of the BEFD-splitting/BESP-splitting schemes 
by considering Example \ref{exmp1} with  $\beta =\delta=10$. 
The initial data is chosen to be \eqref{initial}. 
The problem is solved via BEFD-splitting and BESP-splitting, respectively,  on $[-8,8]$ with $\tau=0.001$. 
The steady state solution is reached when the difference between two consecutive steps are extremely small. In our numerical tests, we stop when 
\be
\frac{\|\Phi^{n+1}-\Phi^n\|_h}{\tau}<10^{-10}.
\ee 
Denote $\phi_{g,h}^{\rm{FD}}$ and $\phi_{g,h}^{\rm{SP}}$ to be the steady states we get via the BEFD-splitting and BESP-splitting, respectively. 
Let $\phi_g$ be the `exact' ground state, which is computed numerically by BESP-splitting with  $h=1/32$ and the same time step $\tau=0.001$. Its corresponding discrete energy, which is denoted as $E_g$, is computed via \eqref{energy:sp}.

\begin{table}
  \centering
  \begin{tabular}{c c c c c }
    \hline
    Error & $h=1/2$ & $h/2$ & $h/2^2$ & $h/2^3$ \\
    \hline
    $|E(\phi_{g,h}^{\rm{FD}})-E(\phi_g)|$ & 9.33E-3 & 2.37E-3 & 5.95E-4 & 1.49E-4 \\
    rate & -& 1.98 & 1.99 & 2.00 \\
    \hline
    $\|\phi_{g,h}^{\rm{FD}}-\phi_g\|_{l_2}$ & 5.14E-3 & 1.30E-3 & 3.23E-4 & 8.06E-5 \\
    rate & -& 1.98 & 2.01 & 2.00  \\
    \hline
     $\|\phi_{g,h}^{\rm{FD}}-\phi_g\|_{\infty}$ & 4.21E-3 & 1.11E-3 & 2.84E-4 & 7.08E-5 \\
    rate & - & 1.93 & 1.96 & 2.01 \\
    \hline
  \end{tabular}
  \caption{Spatial resolution of the ground state for $\beta=\delta=10$ computed by the BEFD-splitting method.}
  \label{tab:spatial_test1}
\end{table}

\begin{table}
  \centering
  \begin{tabular}{c c c c c }
    \hline
    Error & $h=1$ & $h/2$ & $h/2^2$ & $h/2^3$ \\
    \hline
    $|E(\phi_{g,h}^{\rm{SP}})-E(\phi_g)|$ & 1.66E-3 & 1.62E-6 & 1.13E-9 & 5.38E-12 \\
    \hline
    $\|\phi_{g,h}^{\rm{SP}}-\phi_g\|_{l_2}$ & 3.50E-3 & 2.10E-4 & 3.63E-7 & 5.59E-9 \\
    \hline
     $\|\phi_{g,h}^{\rm{SP}}-\phi_g\|_{\infty}$ & 2.02E-3 & 2.02E-4 & 2.53E-7 & 3.88E-9 \\
    \hline
  \end{tabular}
  \caption{Spatial resolution of the ground state for $\beta=\delta=10$ computed by the BESP-splitting method.}
  \label{tab:spatial_test2}
\end{table}

Table \ref{tab:spatial_test1} and Table \ref{tab:spatial_test2} indicate that the BEFD-splitting scheme is second order accurate in space while the BESP-splitting scheme is spectrally accurate. Therefore,  the BESP-splitting scheme requires much less grid points and, therefore, much less computation memory, to get the same order of accuracy. 

\subsection{Ground states in multidimensional space} 
\label{numeric:ground}

In this section, we apply the BEFD-splitting scheme to compute ground states in more general cases. We start with Example \ref{exmp1} 
and compute the ground states for various choices of $\beta$ and $\delta$ to see how the parameters affect the profiles of the ground states. 
The numerical results are shown in Fig. \ref{fig:mgpe_1D_ground}, from where 
we can observe that both the increase of  $\beta$ and $\delta$ will spread the ground state and the phenomenon  would be more obvious if the other parameter is small. 

\begin{figure}[htbp]
\centerline{\psfig{figure=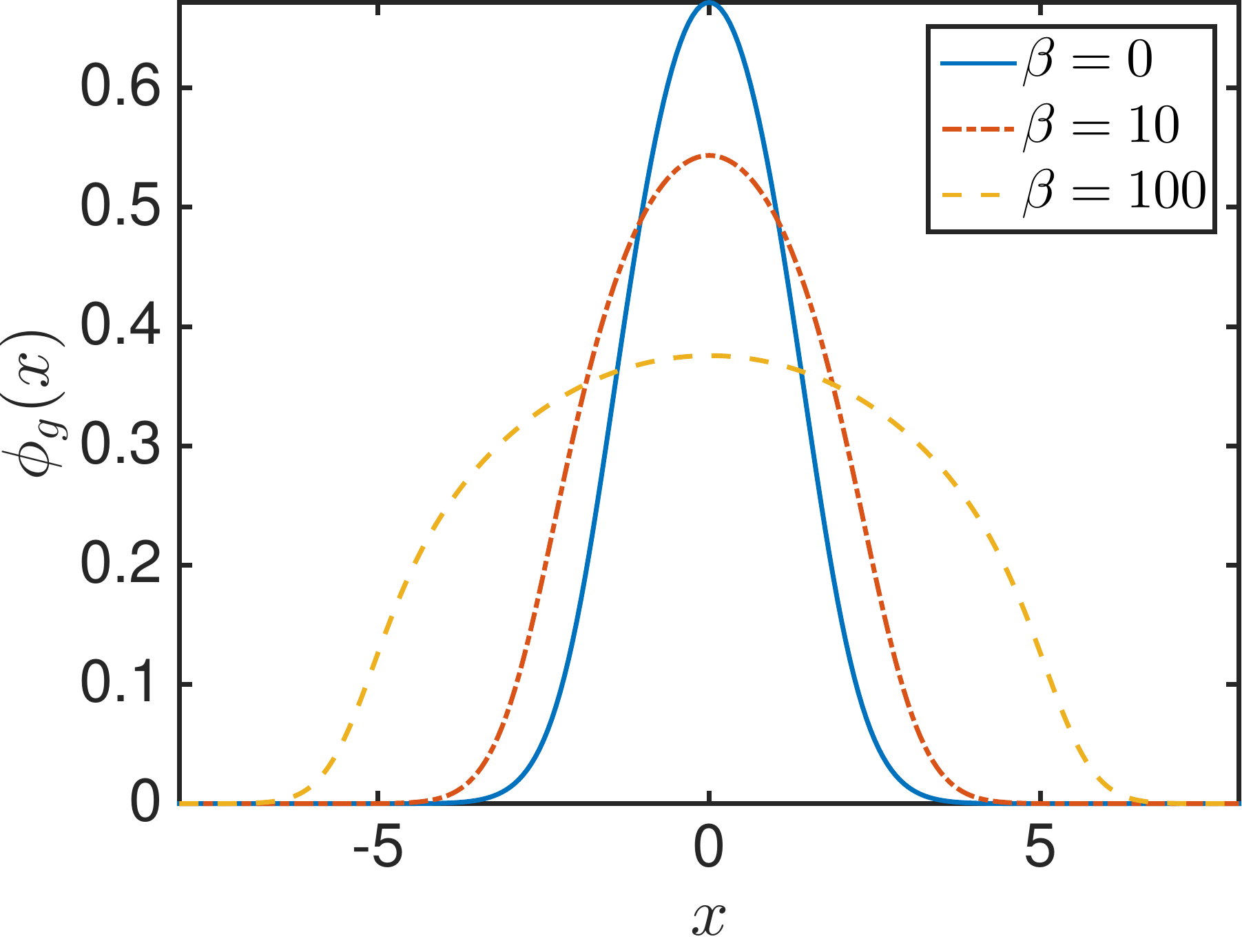,height=\hfig cm,width=\wwfig cm,angle=0}
\psfig{figure=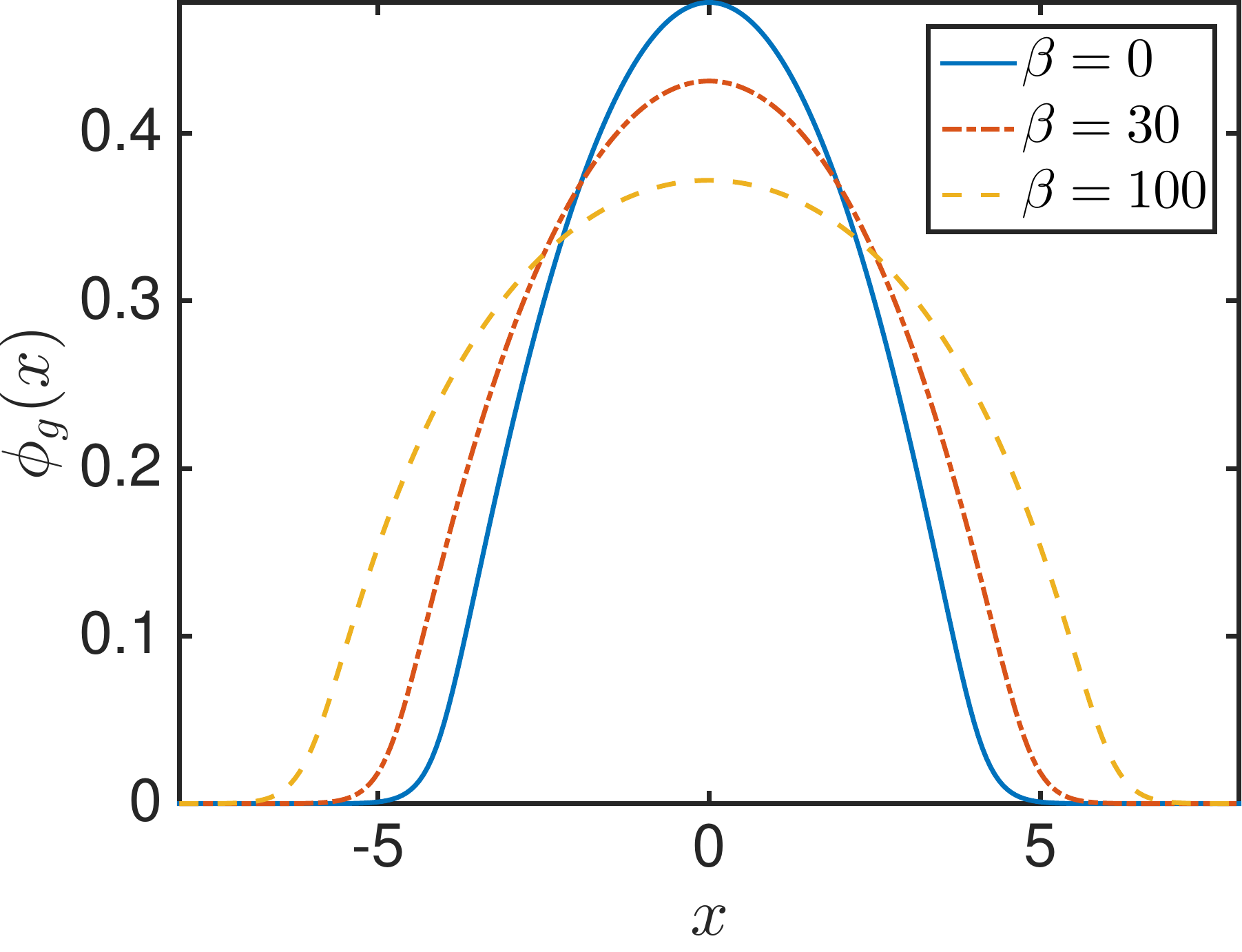,height=\hfig cm,width=\wwfig cm,angle=0}}
\centerline{\psfig{figure=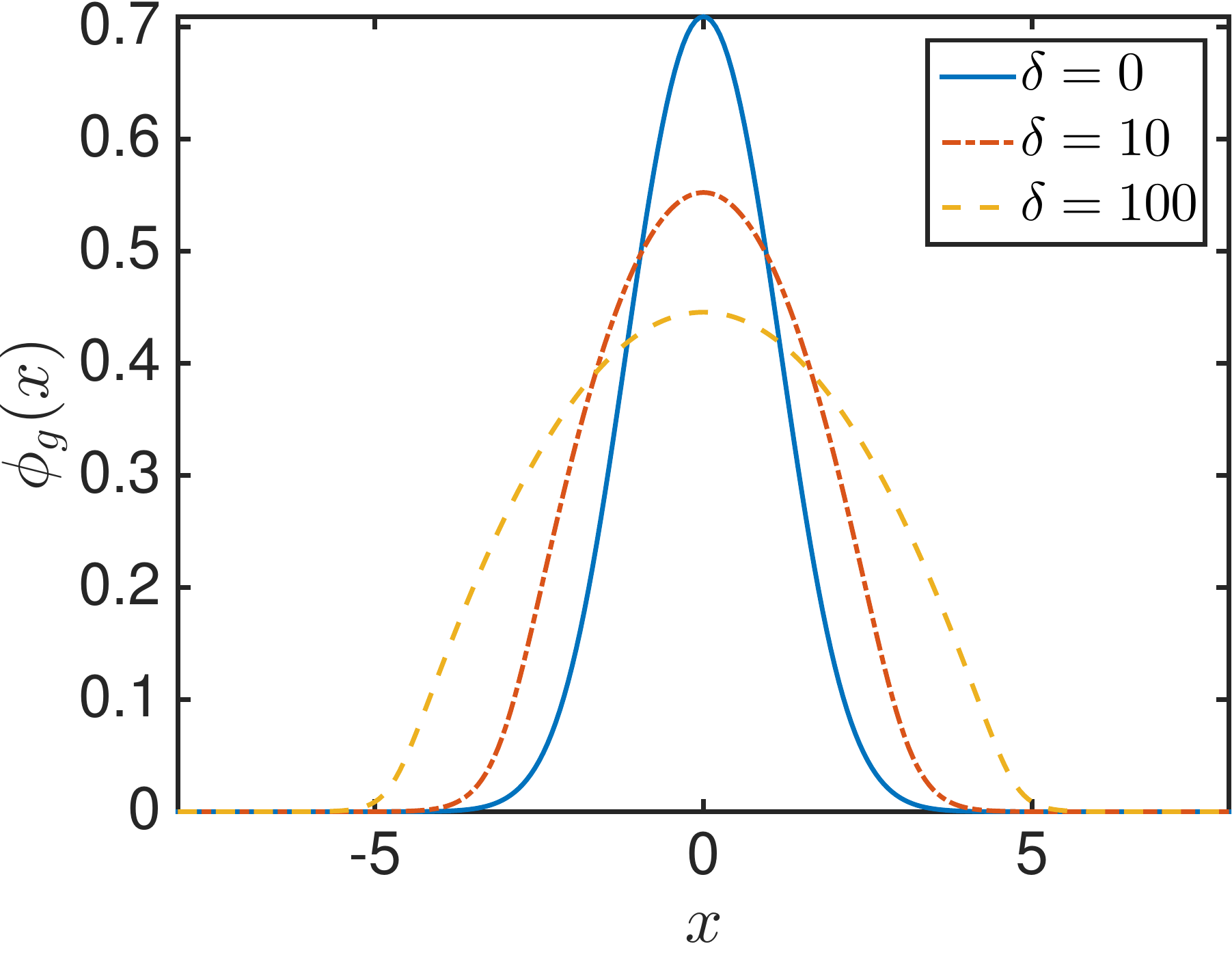,height=\hfig cm,width=\wwfig cm,angle=0}
\psfig{figure=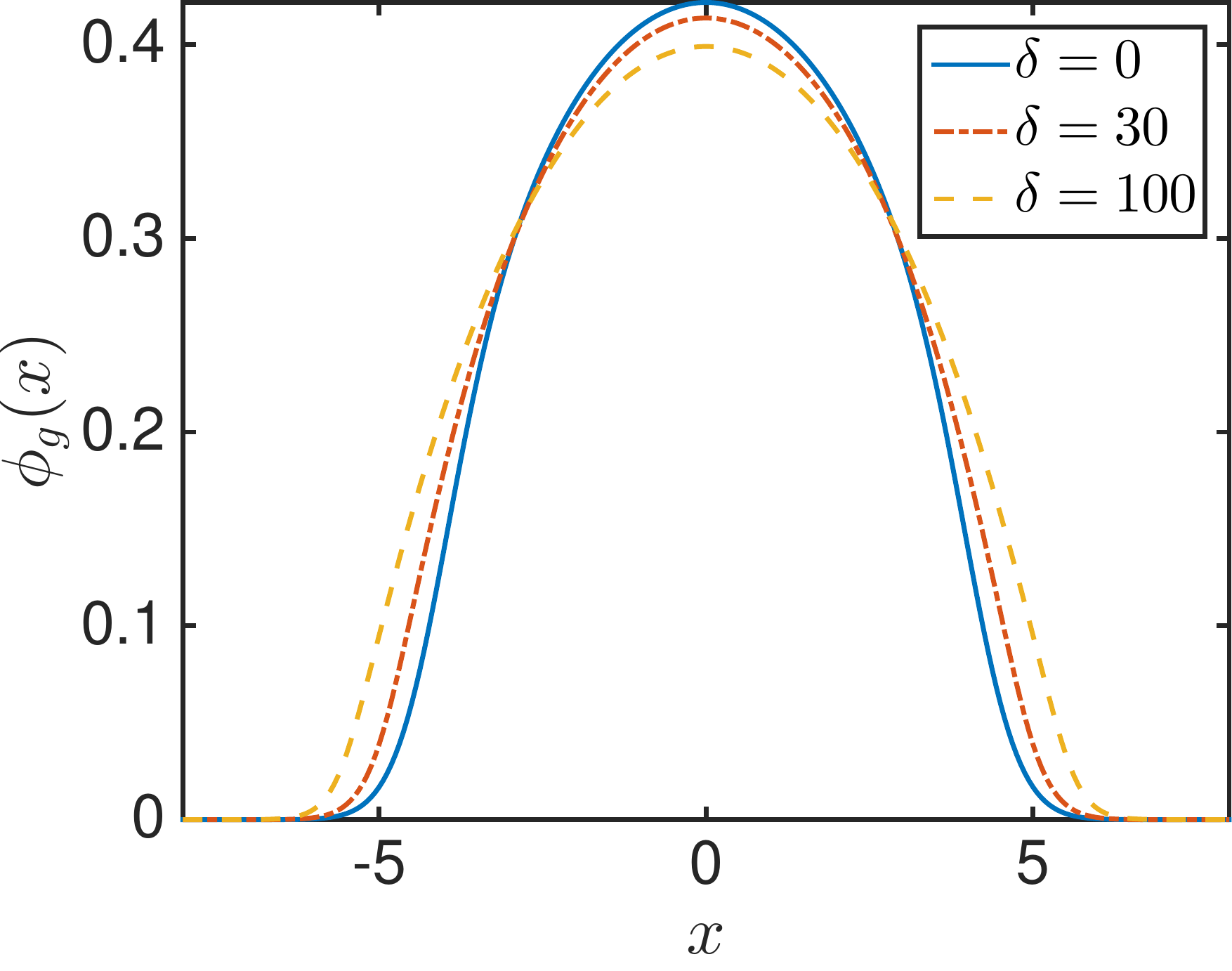,height=\hfig cm,width=\wwfig cm,angle=0}}
\caption{The top row shows the ground states of Example \ref{exmp1} with fixed $\delta=1$ (left) or fixed $\delta=50$ (right) and different $\beta$'s. 
The second row shows the ground states of Example \ref{exmp1} with fixed $\beta=1$ (left) or fixed $\beta=50$ (right) and different $\delta$'s.}
\label{fig:mgpe_1D_ground}
\end{figure}

 The BEFD-splitting method can be applied to multidimensional problems as well. Next we show a case in 2D defined in a bounded domain. 
\begin{example}\label{exmp_2D}
Consider the MGPE under the box potential defined in  $[0,1]^2$, i.e. 
\begin{align}
V(x,y)=\begin{cases}
0, \quad &(x,y)\in[0,1]^2,\\
\infty,\quad &\text{ otherwise }.
\end{cases}
\end{align}
The initial data is chosen to be $\phi_0(x,y)=\sin(\pi x)\sin(\pi y)$.
\end{example}
Obviously, the ground state in Example \ref{exmp_2D} is constrained in $[0,1]^2$ with homogenous Dirichlet boundary conditions. 
By applying the BEFD-splitting scheme, we get the ground states with different choices of  $\beta$ and $\delta$. 
Again, we observe that both the increase of $\beta$ and $\delta$ will make the ground state less concentrated in the center. However, for this problem, the limiting profiles as $\beta\to\infty$ or $\delta\to\infty$ are obviously different. 
If $\beta\gg1$ and $\delta$ is fixed, the ground state is flat and a boundary layer appears. However, such a boundary layer does not appear when $\delta\gg1$ with $\beta$ fixed. 
The results are consistent with the results in \cite{mgpe-th,mgpe-asym,Tho} and the details are omitted here for brevity.

\begin{figure}[htbp]
\centerline{\psfig{figure=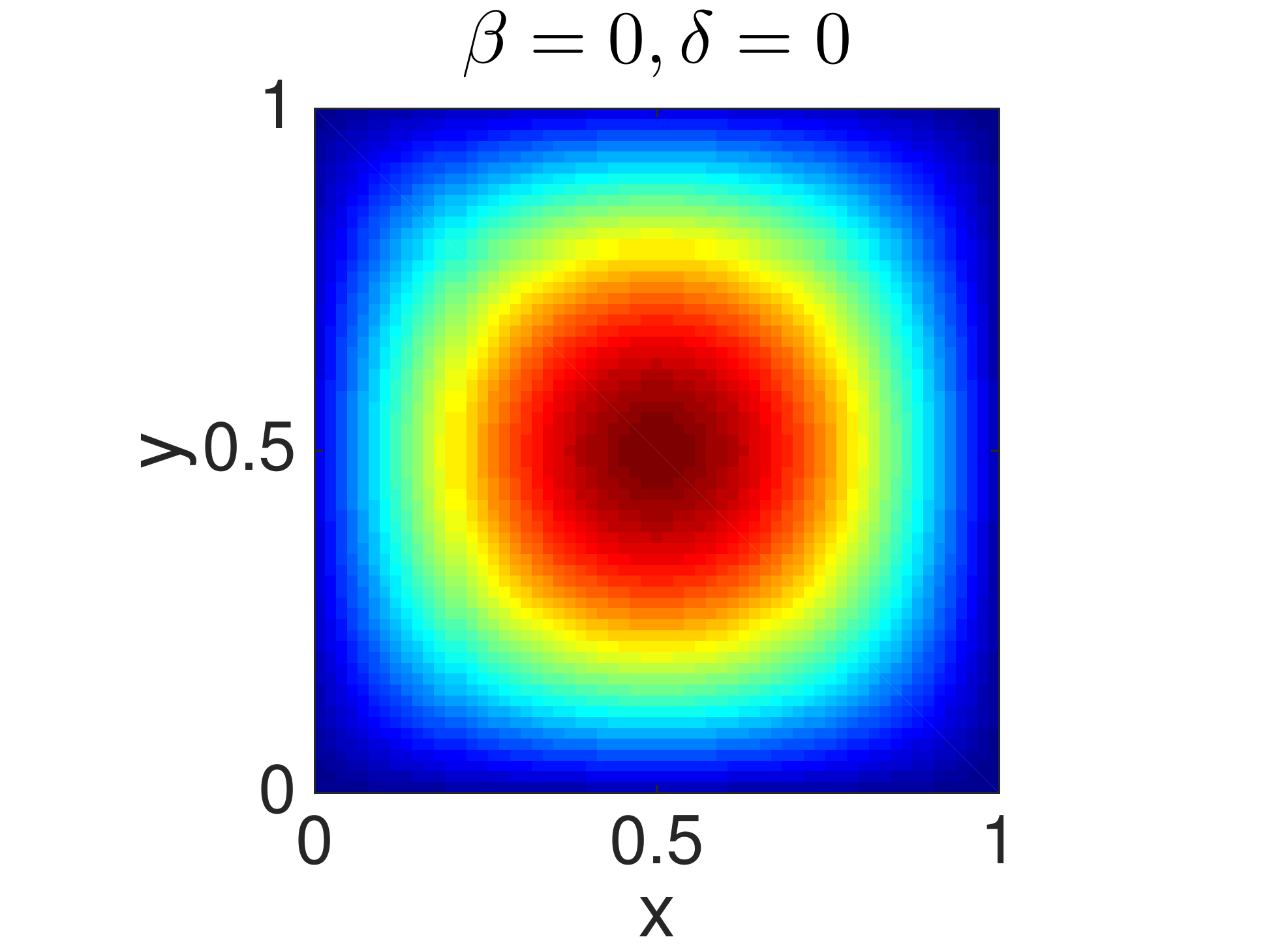,height=\hfig cm,width=\wfig cm,angle=0}
\psfig{figure=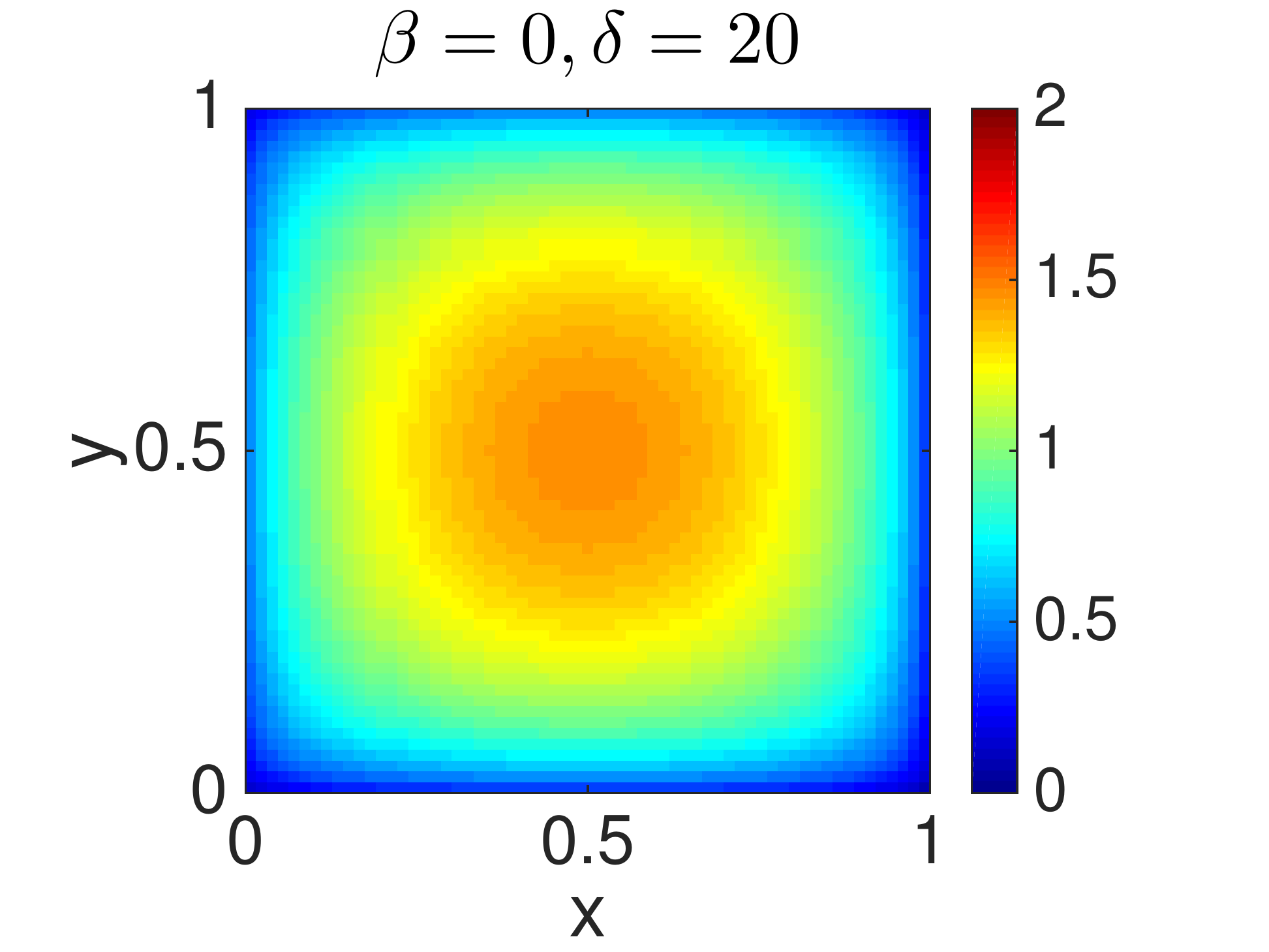,height=\hfig cm,width=\wfig cm,angle=0}}
\centerline{\psfig{figure=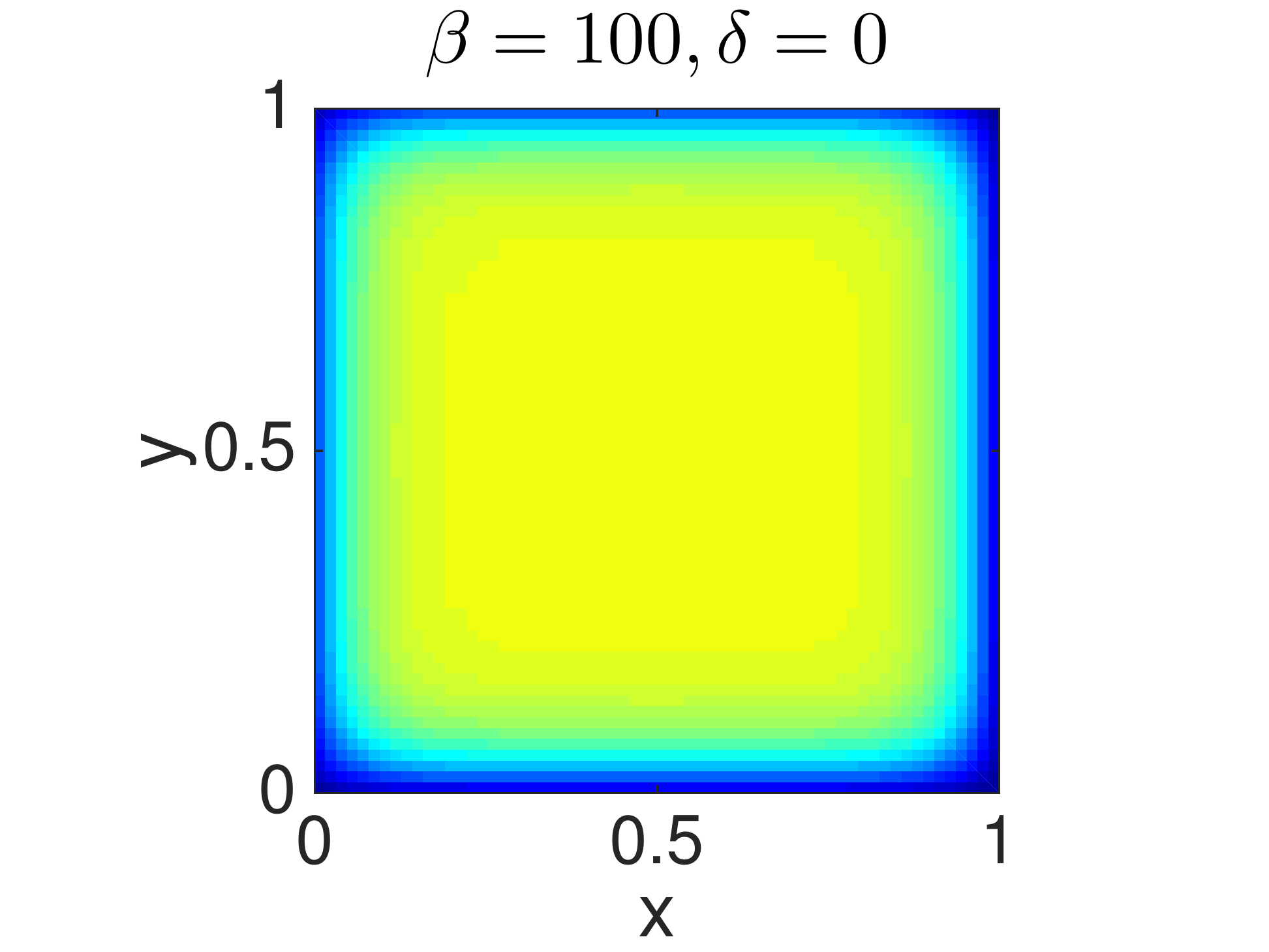,height=\hfig cm,width=\wfig cm,angle=0}
\psfig{figure=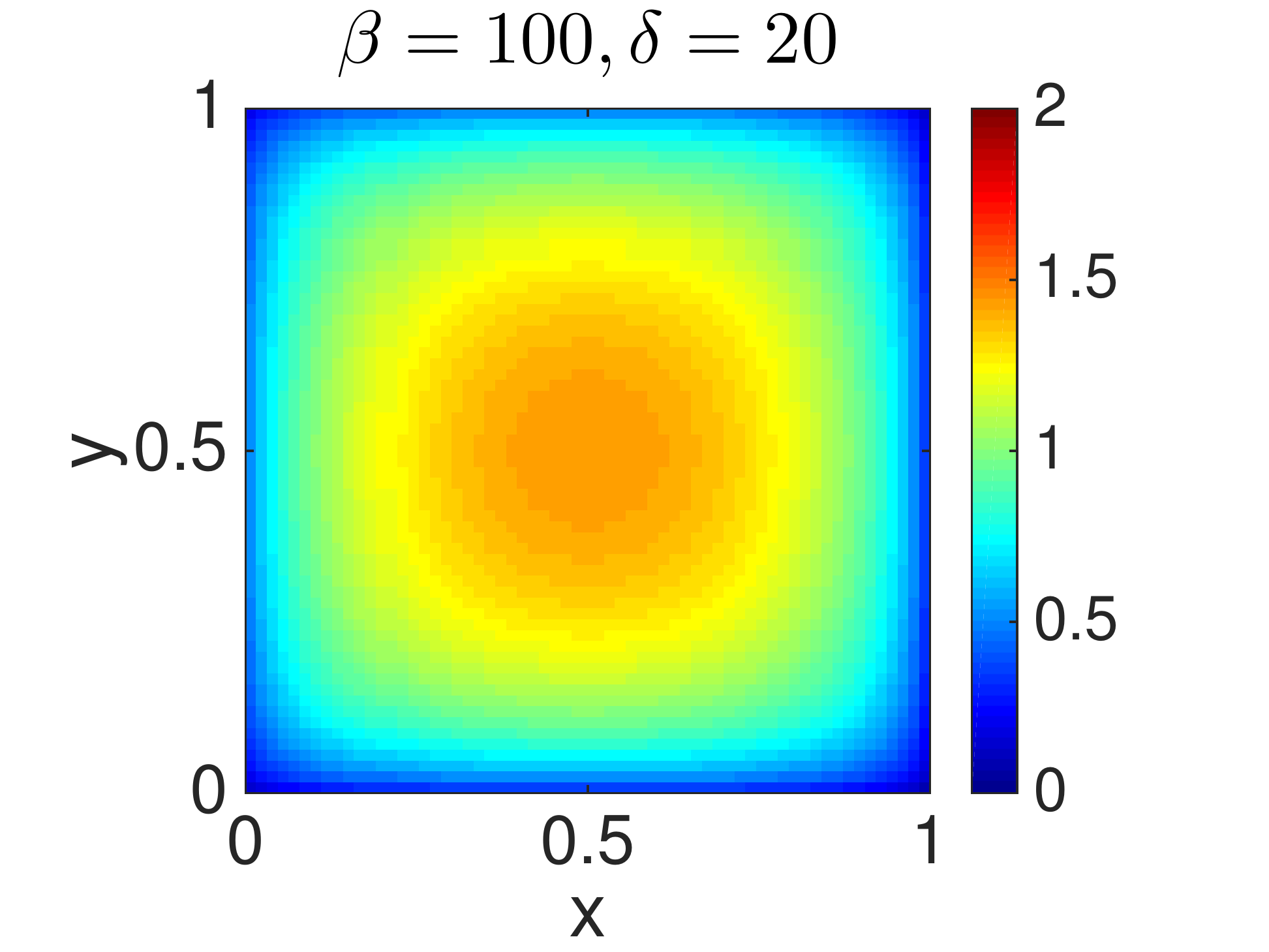,height=\hfig cm,width=\wfig cm,angle=0}}
\caption{Ground states $\phi_g^{\beta,\delta}(x,y)$ with $\beta=0,100$ (from top to bottom) and $\delta=0,20$ (from left to right) under the box potential in $(0,1)^2$.}
\label{fig:mgpe_test1}
\end{figure}

Finally, we  show two numerical examples in 3D, namely Example \ref{exmp2} and Example \ref{exmp3}.  
The BEFD-splitting scheme is applied to compute the ground states in the examples,  
and the isosurfaces are shown in Fig. \ref{fig:mgpe_test2}. 

\begin{example}\label{exmp2}
 Consider the MGPE in 3D under the harmonic potential 
 \be
 V(x,y,z)=(x^2+4y^2+4z^2)/2
 \ee
  with $\beta=1$ and $\delta=20$.
The initial data is chosen to be 
\be\label{exmp2:initial}
\phi_0(x)=\frac{\sqrt{2}}{\pi^{3/4}}e^{-(x^2+2y^2+2z^2)/2}.
\ee
\end{example}
\begin{example}\label{exmp3}
 Consider the MGPE in 3D under the harmonic potential in optical lattices 
 \be
 V(x,y,z)=(x^2+4y^2+4z^2)/2+20(\sin^2(x)+\sin^2(y)+\sin^2(z))
 \ee
  with $\beta=1$ and $\delta=20$ and the initial data \eqref{exmp2:initial}.
\end{example}

\begin{figure}[htbp]
\centerline{\psfig{figure=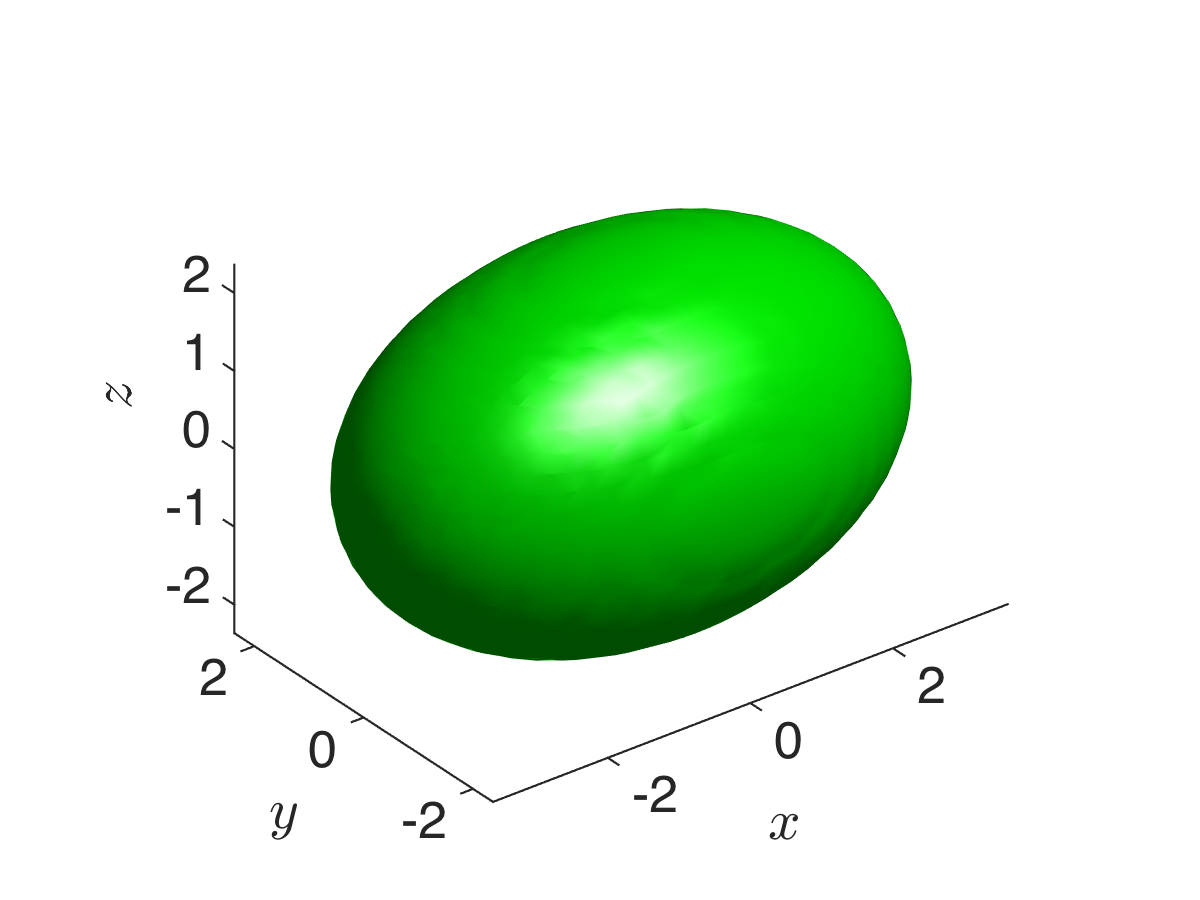,height=\hfig cm,width=\wfig cm,angle=0}
\psfig{figure=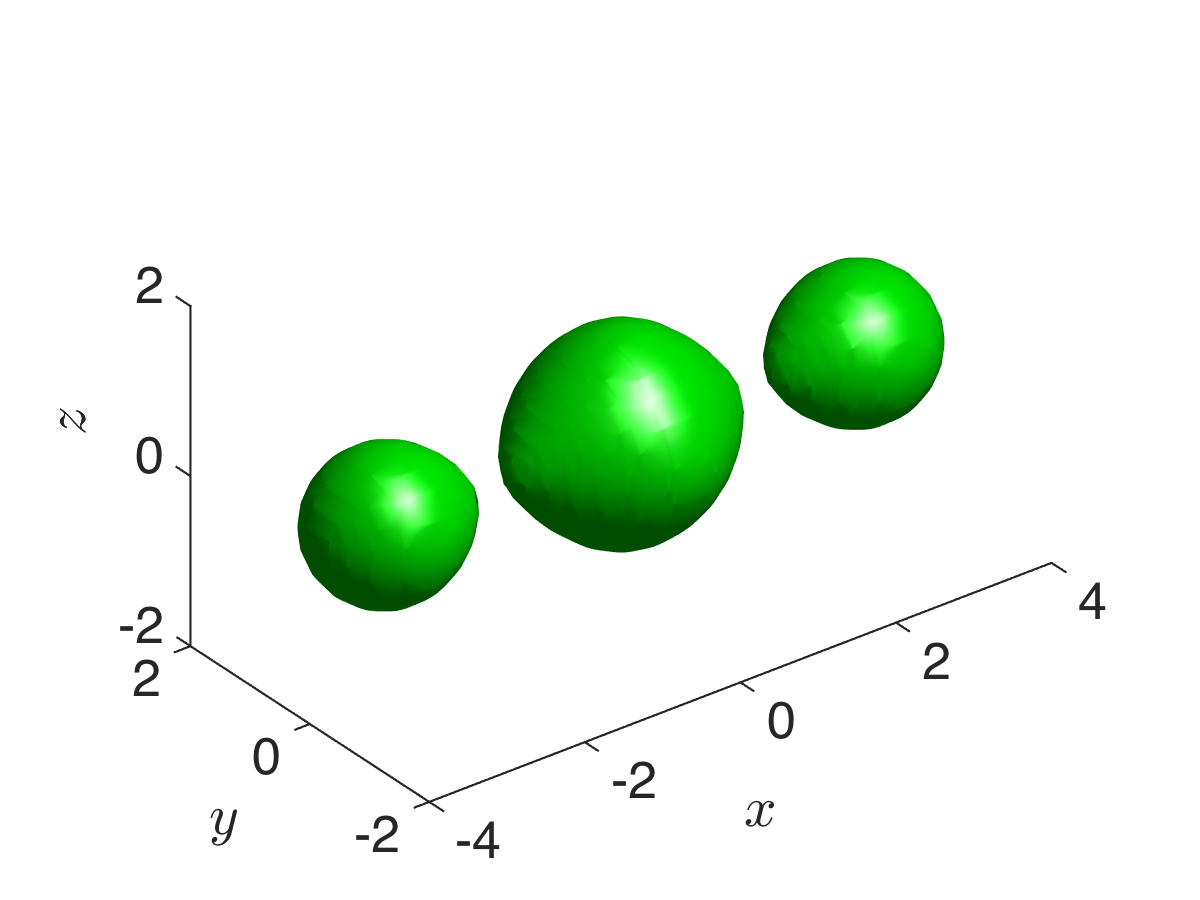,height=\hfig cm,width=\wfig cm,angle=0}}
\caption{Isosurface of the ground states of Example \ref{exmp2} with isovalue 0.01 (left) and Example \ref{exmp3}  with isovalue 0.15 (right). }
\label{fig:mgpe_test2}
\end{figure}

\subsection{Extension to the computation of the first excited state}

For special external potentials, 
the BEFD-splitting/BESP-splitting schemes can also be applied to compute the first excited state, which is the stationary state with the second lowest energy, when the initial data is properly chosen. 
Inspired by the classical GPE case in  \cite{Wz1},   
we can choose the initial data to be the odd function
\be
\phi_0(x)=\frac{\sqrt{2}}{\pi^{1/4}}xe^{-x^2/2}
\ee
to compute the first excited state in Example \ref{exmp1}.
 Fig. \ref{fig:mgpe_1D_1st} shows the first excited states in Example \ref{exmp1} with different choices of $\beta$ and $\delta$. 
For multidimensional problems, the choice of the initial data is similar, and the details can be referred to \cite{Bao2013,gap} and the references therein. 

 \begin{figure}[htbp]
\centerline{\psfig{figure=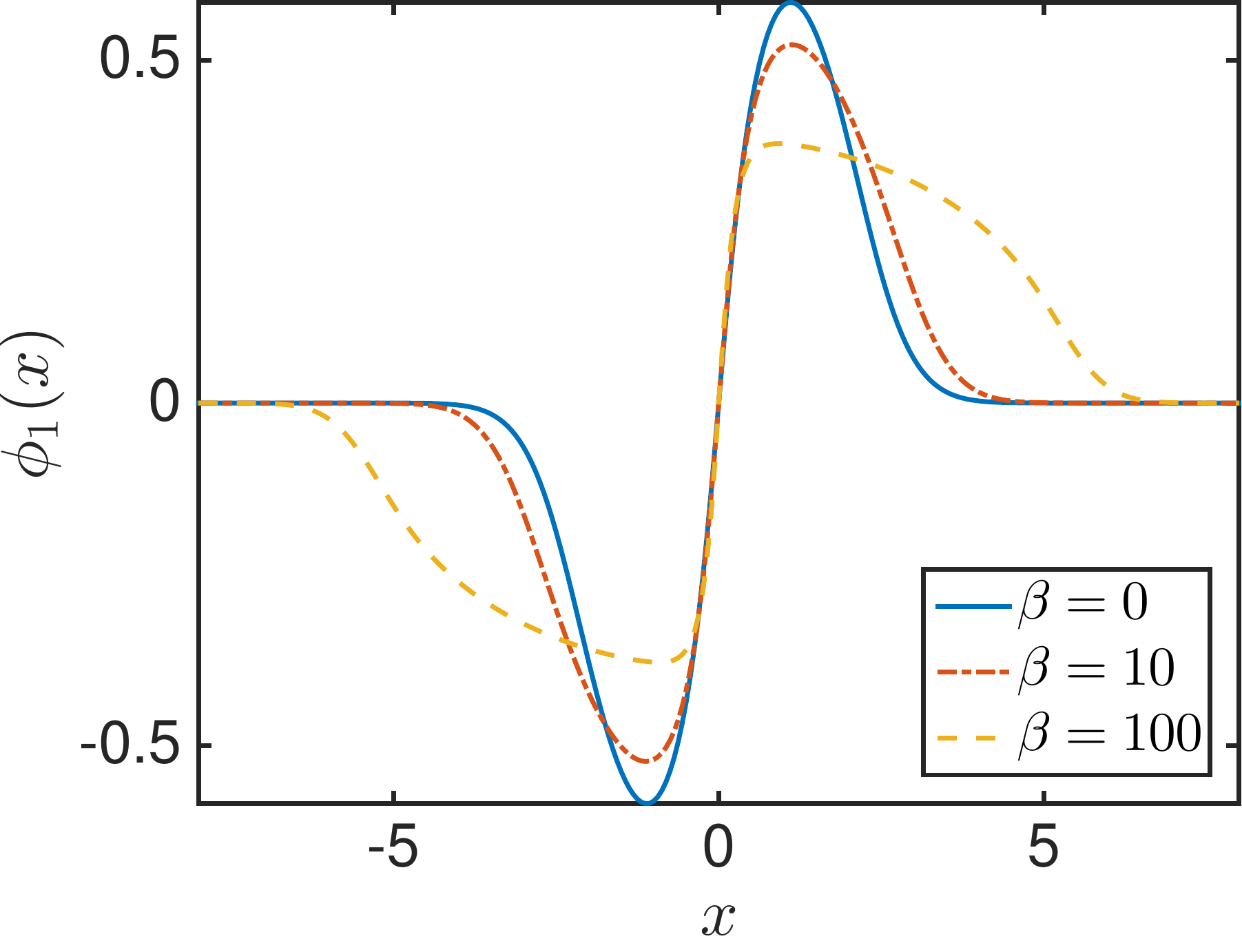,height=\hfig cm,width=\wwfig cm,angle=0}
\psfig{figure=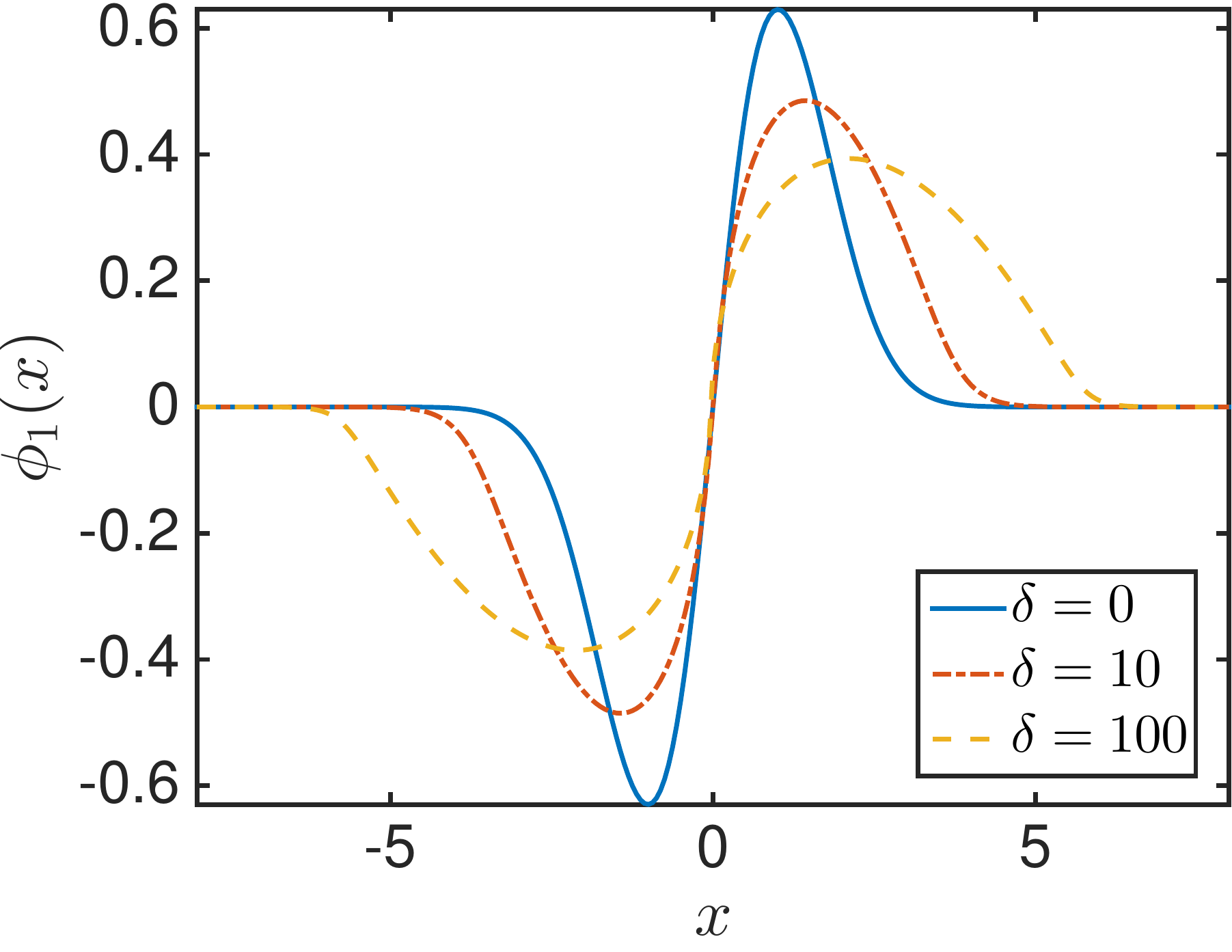,height=\hfig cm,width=\wwfig cm,angle=0}}
\caption{First excited states of Example \ref{exmp1} with fixed $\delta=1$ and different $\beta's$ (left) or fixed $\beta=1$ and different $\delta$'s (right).}
\label{fig:mgpe_1D_1st}
\end{figure}
 
Finally, we show two numerical examples  in 3D, namely Example \ref{exmp2} and Example \ref{exmp3}. 
In both examples, the confinement in the $x$-direction is the weakest. 
Therefore, the first excited state should be the one excited in the $x$-direction, which is similar to the classical GPE case where the first excited state is studied in details in \cite{asym,gap}. 
The first excited state can be computed by the BEFD-splitting method with  
the initial data 
\be
\phi_0(x)=\frac{2x}{\pi^{3/4}}e^{-(x^2+2y^2+2z^2)/2}
\ee
for both examples, 
and the numerical results are shown in Fig. \ref{fig:mgpe_test2}.

\begin{figure}[htbp]
\centerline{
\psfig{figure=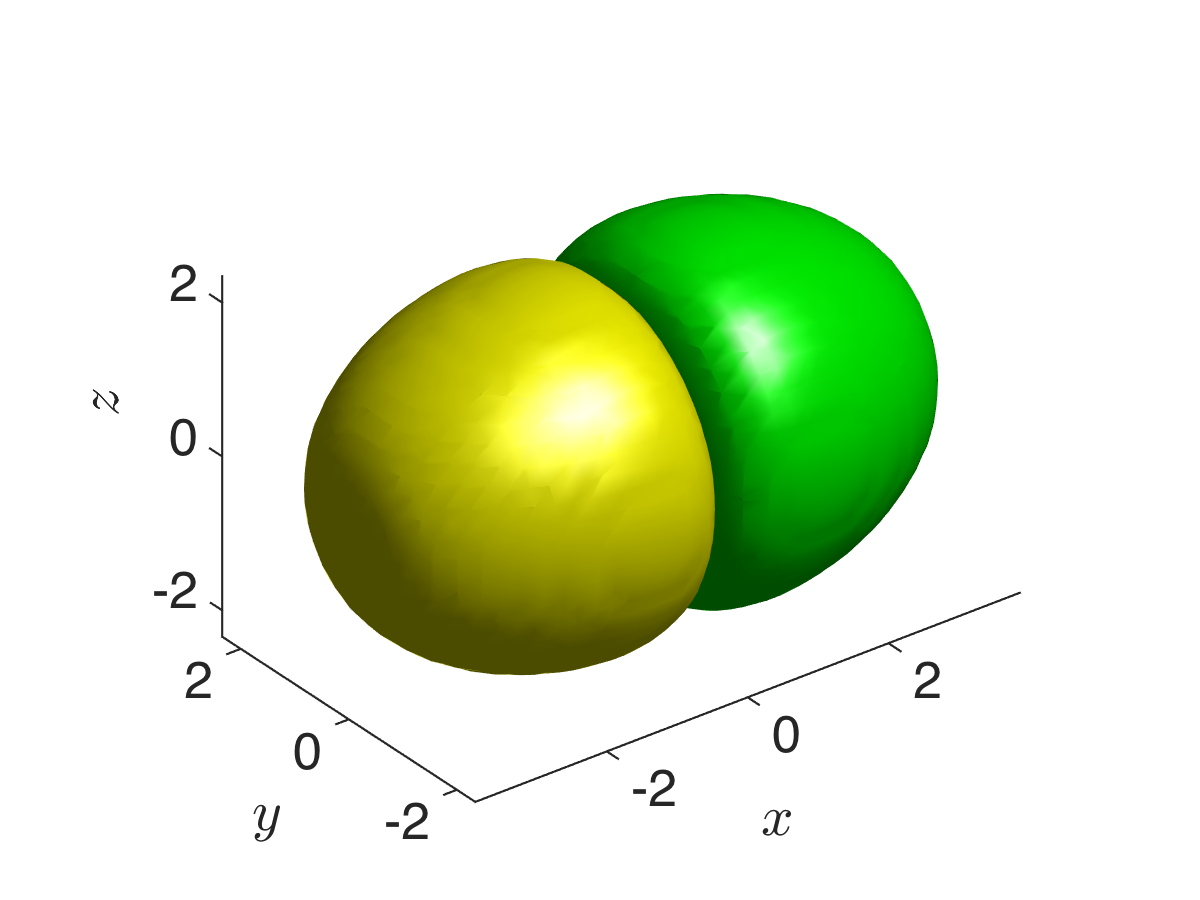,height=\hfig cm,width=\wfig cm,angle=0}
\psfig{figure=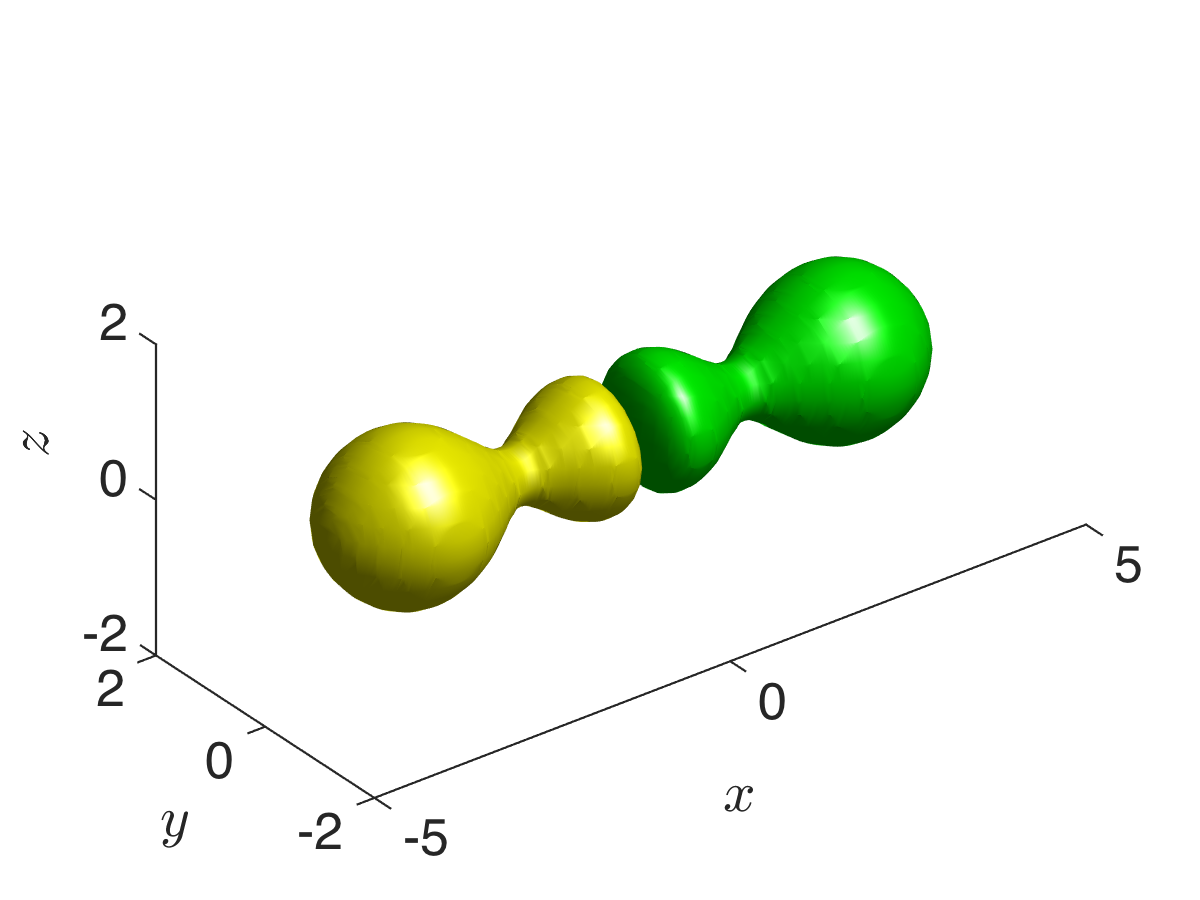,height=\hfig cm,width=\wfig cm,angle=0}}
\caption{First excited states of the MGPE in 3D for Example \ref{exmp2}  with isovalues $\pm0.01$ (left) and Example \ref{exmp3}  with isovalues  $\pm0.15$ (right). The green part is for the part with positive value while the yellow part is for the part with negative value.}
\label{fig:mgpe_test3}
\end{figure}

\section{Conclusion}\label{sec:conclusion}
In this paper, we generalized the normalized gradient flow method, which was originally designed for GPE, to compute the ground state of the MGPE \eqref{mgpe}. 
In particular, the CNGF-FD scheme \eqref{CNGF-FD} proposed can be proven to be normalization conservative and energy diminishing. However, the scheme is not suitable for practical numerical computation since a complicated nonlinear equation needs to be solved for each step.  

To design an easy, efficient and stable numerical scheme suitable for practical numerical computation, we introduced the attractive-repulsive splitting of the term $\delta\Delta(|\phi|^2)\phi$ and constructed the BEFD-splitting/BESP-splitting schemes, which are explicit-implicit schemes and only a linear equation is needed to be solved for each step. 
Numerical experiments indicate that the new schemes are much more stable than schemes constructed in naive ways and are competitive with other popular methods.  
With the multigrid technique and a relatively large time step, the scheme can be applied to MGPE with extremely strong nonlinearities, which implies that the BEFD-splitting/BESP-splitting schemes are extremely suitable for computing  the ground state of the MGPE. 

\section*{Acknowledgments}
This work was supported by the Academic Research
Fund of Ministry of Education of Singapore grant No.
R-146-000-223-112 (MOE2015-T2-2-146) and I would like to specially thank Prof.  Weizhu Bao for his stimulating discussion.

\end{document}